\documentclass[a4paper,11pt]{article} 
  
\usepackage[english]{babel}
\usepackage[a4paper]{geometry}
\usepackage{enumerate}  
\usepackage{amsmath}
\usepackage{graphicx}
\usepackage{amsthm}
\usepackage{amssymb}
\usepackage{hyperref} 
\usepackage{bbm}
\usepackage{xcolor}
\usepackage{amssymb}
\usepackage{color}  
\usepackage{hyperref}

\def\bR{\mathbb{R}}

\def\bT{\mathbb{T}}
\def\bN{\mathbb{N}}

\def\bZ{\mathbb{Z}}
\def\cN{\mathcal{N}}
\def\cQ{\mathcal{Q}}
\def\cF{\mathcal{F}}
\def\cB{\mathcal{B}}
\def\cE{\mathcal{E}}

\def\cH{\mathcal{H}}
\def\eps{\varepsilon}
\def\ph{\varphi}

\def\cL{\mathcal{L}}

\def\cS{\mathcal{S}}
\def\cD{\mathcal{D}}
\def\cG{\mathcal{G}}

\def\tbf{\textbf }

\def\wh{\widehat}
\def\be{\begin{equation}}
\def\ee{\end{equation}}
\def\wh{\widehat}

\def \pn { \varphi_\text{GP} }

\def\eps{\varepsilon}
\def\ph{\varphi}

\def\wh{\widehat}
\def\wch{\check}

\allowdisplaybreaks
\newtheorem{theorem}{Theorem}[section]

\newtheorem{prop}[theorem]{Proposition}
\newtheorem{lemma}[theorem]{Lemma}
\newtheorem{proposition}[theorem]{Proposition}
\newtheorem{definition}[theorem]{Definition}
\newtheorem{corollary}[theorem]{Corollary}
\newtheorem{assumption}[theorem]{Assumption}

\theoremstyle{definition} 						
\newtheorem{remark}[theorem]{Remark}			
\numberwithin{equation}{section}

\begin{document}
\title{Exponential Control of Excitations for Trapped BEC in the Gross-Pitaevskii Regime}
\author{Nils Behrmann\thanks{Institute for Applied Mathematics, University of Bonn, Endenicher Allee 60, 53115 Bonn, Germany}  \and  Christian Brennecke\footnotemark[1] \and Simone Rademacher\thanks{Department of Mathematics, LMU Munich, Theresienstr. 39, 80333 Munich,
Germany}}

\maketitle

\begin{abstract}
We consider trapped Bose gases in three dimensions in the Gross-Pitaevskii regime whose low energy states are well known to exhibit Bose-Einstein condensation. That is, the majority of the particles occupies the same condensate state. We prove exponential control of the number of particles orthogonal to the condensate state, generalizing recent results from \cite{NR} for translation invariant systems. 
\end{abstract}

\section{Introduction and Main Result} 
\label{sec:intro}

\subsection{Introduction}

Bose gases stand out for a special property: at low temperature, they perform a phase transition in which the majority of the Bose gas' particles condense into the same quantum state, called Bose-Einstein condensate. Since its theoretical prediction by Bose and Einstein \cite{Bose,Einstein} and its first experimental observations \cite{WC,K}, it is of great interest to understand this phenomenon mathematically. The goal of the present paper is to contribute to this understanding and our main result is to prove exponential control of the number of particles that are not in the Bose-Einstein condensate, the so called orthogonal excitations, in the Gross-Pitaevskii limit.  Our result generalizes earlier results from \cite{NR} for translation invariant systems on the unit torus to physically more realistic systems of trapped Bose gases in $\mathbb{R}^3$, modelling experimental setups as in \cite{WC,K}.  
 \\
 
The state of $N$ bosons in $\bR^3$ is described by a normalized vector in $L_s^2( \mathbb{R}^{3N})$, the symmetric wave functions in $L^2( \mathbb{R}^{3N})$. We describe the energy by the Hamiltonian 
		\begin{align}
		\label{def:ham}	H_N = \sum\limits_{j=1}^N \big(-\Delta_{x_j} +V_{\text{ext}}(x_j) \big)+\sum\limits_{1\leq i < j\leq N} N^2V(N(x_i-x_j))  . 
		\end{align}
Its characteristic is that the two-body interaction potential $N^2V(N \cdot)$ scales with the total number of particles $N$ of the system and the scaling is chosen such that its scattering length $\mathfrak{a}$ is of order $N^{-1}$: The potential $V$'s scattering length $\mathfrak{a}$ is defined by 
\begin{align}
\label{def:a0}
8 \pi \mathfrak{a} = \int_{\mathbb{R}^3} dx\,V(x) f(x),
\end{align} 
where $f$ denotes the solution to the zero-energy scattering equation 
\begin{align}
\label{eq:scattering}
\big(- 2 \Delta +  V  \big) f  =0  
\end{align}
with boundary condition $f(x) \rightarrow 1$ as $\vert x \vert \rightarrow \infty$. Rescaling \eqref{def:a0},  eq. \eqref{eq:scattering} implies that the scattering length of the scaled potential $N^2 V(N \cdot)$ is $\mathfrak{a}N^{-1}$. 

As shown in \cite{LSY, LS1}, systems as in \eqref{def:ham} admit an effective description through the Gross-Pitaevskii theory: in the limit $N\to\infty$, the leading order contribution to the ground state energy $E_N$ of $H_N$ is described by the mimimum 
		\[  \lim_{N\to\infty} N^{-1}E_N = \inf_{\|\ph\|=1} \mathcal{E}_{\text{GP}} ( \ph ) \]
of the Gross-Pitaevskii energy functional $ \mathcal{E}_{\text{GP}}$, defined by
		\begin{align} \label{def:EGP}
		\mathcal{E}_{\text{GP}} ( \ph ) =   \int_{\mathbb{R}^3} dx \, \big( \vert \nabla \ph  \vert^2 + V_{\text{ext}} \vert \ph  \vert^2 + 4 \pi \mathfrak{a} \vert \ph  \vert^4 \big). 
		\end{align}
Moreover, the normalized ground state $\psi_N$ of $H_N$ exhibits complete BEC into the unique normalized, strictly positive minimizer $\varphi_{\text{GP}}$ of $\mathcal{E}_{\text{GP}} $. Mathematically, this means that the number of particles orthogonal to the condensate, counted by the operator
	\begin{align}
	\label{def:N}
	\cN_{\bot \ph_\text{GP} } = \sum_{i=1}^N Q_i, \quad \text{where} \quad Q = 1- \vert \varphi_{\text{GP}} \rangle \langle \varphi_{\text{GP}} \vert,
	\end{align}
is negligible with regards to the total number $N$ of particles, in the limit $N\to\infty$. Here,  $Q_i = {\bf 1} \cdots  \otimes {\bf 1} \otimes Q \otimes  {\bf 1} \cdots  {\bf 1}$ is the operator that acts on the $i$-th particle as the operator $Q$ and as identity on the remaining particles. The main result of \cite{LS1} is that
		\be \label{eq:BEC}  \lim_{N\to\infty} N^{-1} \langle \psi_N, \cN_{\bot \ph_\text{GP} } \psi_N\rangle = 0.  \ee
This was generalized in \cite{LS2} and later, with different tools, in \cite{NRS} to approximate ground states. More recent work \cite{NNRT,BSS1} (see also \cite[Appendix D]{BK}) improved \eqref{eq:BEC} quantitatively by showing that $ \langle \psi_N, \cN_{\bot \ph_\text{GP} } \psi_N\rangle = O(1)$, which is optimal. In fact, it turns out that 
		\be \label{eq:BECpoly}  \langle \psi_N, \cN_{\bot \ph_\text{GP} }^k \psi_N\rangle = O(1)\ee 
for every fixed $k\in\bN$. This was proved recently in \cite{BSS2} in the context of the derivation of Bogoliubov's effective theory, generalizing previous work for translation invariant systems \cite{BBCS1, BBCS2, BBCS3, BBCS4, ABS}; see also \cite{HST, NT,Br}. 

Viewing $\cN_{\bot \ph_\text{GP} }$ as a random variable with regards to the spectral measure induced by $\psi_N$, the moment bound \eqref{eq:BECpoly} implies that the probability of finding more than $n>0$ excited particles in the ground state decays at least with a rate $1/n^k$, as $n\to\infty$ (see also \tbf{Remarks 1.5} \& \tbf{1.6} below). Our main result improves this and shows that this probability decays exponentially in $n$, based on an exponential moment bound on $\cN_{\bot \ph_\text{GP} }$. Such bounds are relevant for the analysis of the large deviations of $\cN_{\bot \ph_\text{GP} }$ (viewed as the sum of the $N$ identically distributed, weakly dependent random variables $Q_i$) recently studied in other settings in \cite{ NR, R-ldp, R}. 

\subsection{Main Result}
Our main result is proved within the framework of \cite{BSS1, BSS2}. In particular, we make use of the main result of \cite{BSS1} and this requires us to set up the following assumptions: 
		\begin{assumption}  \label{ass:V}
		We assume that $V, V_{\emph{ext}}: \mathbb{R}^3 \rightarrow \mathbb{R}$ satisfy:
		\begin{enumerate}
		\item[(i)] $V \in L^3( \mathbb{R}^3)$ is radial, compactly supported and $V(x) \geq 0$ for a.e. $x \in \mathbb{R}^3$,
		\item[(ii)] $V_{\emph{ext}} \in C^1( \mathbb{R}^3)$ is such that $V_{\emph{ext}} (x) \rightarrow \infty$ as $\vert x \vert \rightarrow \infty$, $\nabla V_{\emph{ext}}$ has at most exponential growth as $ |  x  | \rightarrow \infty$ and there exists a constant $ C>0$ such that $ V_{\emph{ext}} (x+y) \leq C \big( V_{\emph{ext}} (x) + C\big) \big( V_{\emph{ext}} (y) + C \big) $ for every $x,y\in \bR^3$. 
		\end{enumerate}
		\end{assumption} 
Based on Assumption \ref{ass:V}, it was proved in \cite[Theorem 1.1]{BSS1} that for some $C>0$ 
		\be \label{eq:coerc} H_N \geq N  \mathcal{E}_{\text{GP}}(\ph_\text{GP}) + C^{-1}\cN_{\bot \ph_\text{GP} } - C.  \ee
The bound \eqref{eq:coerc} is an essential ingredient in the proof of our following main result. 
		\begin{theorem} \label{thm:main}
		Let $V, V_{\emph{ext}}$ satisfy Assumption \ref{ass:V}, let $\zeta>0$ and denote by $E_N$ the ground state energy of $H_N$, defined in \eqref{def:ham}. Then, there exist a constant $C>0$ such that for all normalized vectors $\psi_N = {\bf 1}_{ (-\infty, E_N + \zeta] } (H_N) \psi_N $, $\| \psi_N \| =1$, in the spectral subspace of $H_N$ with energies below $E_N + \zeta$ and for all $\kappa>0$ sufficiently small, we have 
		\begin{align}\label{eq:exp}
		\langle \psi_N, \exp ( \kappa \cN_{\bot \ph_\emph{GP}}  ) \psi_N \rangle \leq C. \end{align}
		\end{theorem}
		
\begin{remark}[Moment and Exponential Moment Bounds] 
The proof of \eqref{eq:exp} is based on a refinement of a commutator argument from \cite{BBCS2, BBCS3, BSS2} that implies uniform moment control \eqref{eq:BECpoly} on low-energy spectral subspaces of $H_N$. One possibility to proceed here is to refine the moment bounds \eqref{eq:BECpoly} from \cite{BSS2} to bounds of the form
		 \[  \langle \psi_N, \cN_{\bot \ph_\text{GP} }^k \psi_N\rangle \leq C_\zeta^k k!\]
for some constant $C_\zeta$ that is independent of $N$ and $\psi_N$ (cf.\ also \cite[Remark 1.1]{NR}). This was recently implemented in detail in the master thesis \cite{Behr}. In this paper, we present a simplified and streamlined version of the method. Our proof combines the analysis of \cite{BSS1, BSS2} with a similar argument as in \cite{NR} which proves bounds like \eqref{eq:exp} for translation invariant systems in the Gross-Pitaevskii limit. Compared to \cite{NR}, whose bounds are proved for normalized eigenfunctions $\psi_N$ of $H_N$ with low energies $E_N + \zeta$ for some $\zeta >0$, our bounds are proved directly in the spectral subspace of $H_N$ with energies up to $E_N + \zeta$. To illustrate our main argument in a simple context, we briefly outline the proof of an analogue of Theorem \ref{thm:main} for translation invariant mean field systems in Section \ref{sec:outline} (for a direct comparison of our method with that of \cite{NR}, cf.\ \cite[Section 1.2]{NR}). This gives an alternative proof of a comparable result proved previously in \cite{M} (see also \cite{MP} for inhomogeneous mean field systems).

\end{remark}

\begin{remark}[Moment generating function] 
For the translation invariant setting on the unit torus, the exponential bounds in \cite{NR} of the form \eqref{eq:exp} were further improved: recently, \cite{R} established an explicit expression for the moment generating function of the number of excitations $\cN_{\bot \ph_\text{GP} }$ in the large particle limit. More precisely, \cite{R} shows that there exists $\lambda_0 > 0$ such that 
\begin{align}
\label{eq:gf1}
\lim_{N \rightarrow \infty} \langle \psi_N , \exp( \lambda \cN_{\bot \ph_\text{GP} }) \psi_N \rangle = e^{\Lambda ( \lambda)} 
\end{align}
for all $\vert \lambda \vert \leq \lambda_0$. The convex rate function  $\Lambda : (-\lambda_0, \lambda_0) \rightarrow \mathbb{R}$ is explicitly  given in terms of the scattering length $\mathfrak{a}$ and moreover, serves as a moment generating function, i.e. 
\begin{align}
\label{eq:gf2}
\lim_{N \rightarrow \infty} \langle \psi_N , \cN_{\bot \ph_\text{GP} }^k \psi_N \rangle  =  \frac{d^k}{d\lambda^k} e^{\Lambda ( \lambda )} \bigg\vert_{\lambda =0} \; . 
\end{align} 
Thus, asymptotically any moment of $\cN_{\bot \ph_\text{GP} }$ can be computed through taking derivatives of $\Lambda$, in particular recovering back earlier results on asymptotic formulas for the first \cite{BBCS3} and, respectively, second moment \cite{NR}.

The arguments in \cite{R} build on a rigorous version of Bogoliubov's theory \cite{BBCS3} that is technically easier to formulate and verify on the unit torus compared to the trapped setting considered in this paper. Based on \cite{NT,BSS2}, however, we expect similar results as \eqref{eq:gf1}, \eqref{eq:gf2} in the trapped setting, too.

\end{remark}

\begin{remark}[Probabilistic interpretation] From the first principles of quantum mechanics, a normalized wave function $\psi_N \in L_s^2( \mathbb{R}^{3N})$ defines a probability distribution: similarly as in \eqref{def:N}, we define for a self-adjoint one-particle operator $O$ on $L^2( \mathbb{R}^3)$ the $N$-particle operator $O_i$ as the operator acting as $O$ and the $i$-th particle and as identity on the remaining $N-1$ particles. Then by the spectral calculus, the operator $O_i$ gives rise to a random variable with law 
\[
\mathbb{P}_{\psi_N} \big[ O_i \in A \big] = \langle \psi_N , {\bf 1}_{A} ( O_i ) \psi_N \rangle \quad \text{for every} \quad A \in \cB(\mathbb{R}) \; . 
\]
Within the framework of this probabilistic approach, first studied in the dynamical setting in \cite{BKS}, central limit theorems for bosons in Gross-Pitaevskii regime were proven for the equilibrium \cite{RS,PRV} and the dynamics \cite{COS}.  

In the probabilistic interpretation, Theorem \ref{thm:main} implies an upper bound on the large deviations for the random variable $\mathcal{N}_{\bot\ph_\text{GP}}  =\sum_{i=1}^N Q_i $: Markov's inequality and Theorem \ref{thm:main} show that there exists $\kappa >0$ such that for everty $0 \leq n \leq N$ and $\psi_N = {\bf 1}_{(-\infty, E_N + \zeta]} \psi_N $
\begin{align}
\mathbb{P}_{\psi_N} \big[ \cN_{\bot \ph_\text{GP} } \geq n \big] \leq e^{- \kappa n} \; . \label{eq:ld}
\end{align}
In other words , Theorem \ref{thm:main} shows that the probability of finding $n$ particles outside the condensate decays exponentially in $n$. 
\end{remark}

\begin{remark}[Large deviations]
For the Bose gas on the unit torus, the decay rate of \eqref{eq:ld} can be further refined by using the asymptotic formula \eqref{eq:gf1} for the moment generating function. Based on \cite{R}, it follows that
\begin{align}
\liminf_{N \rightarrow \infty} \log \mathbb{P}_{\psi_N} \big[ \cN_{\bot \ph_\text{GP} }  \geq n \big] \leq \inf_{0 \leq \lambda \leq \lambda_0} \big[ - \lambda  n   + \Lambda ( \lambda) \big]  \;, \label{eq:ldp1}
\end{align}
and expanding the r.h.s.\ for small $\lambda >0$ shows that 
\begin{align}
\label{eq:ldp2}
\liminf_{N \rightarrow \infty} \log \mathbb{P}_{\psi_N} \big[ \cN_{\bot \ph_\text{GP} }  \geq n \big] \leq \inf_{0 \leq \lambda \leq \lambda_0} \big[ - \lambda (n  - \mu ) + \frac{\sigma^2\lambda^2}{2}  + O( \lambda^3) \big]  \; . 
\end{align}
where $\mu = \lim_{N \rightarrow \infty} \langle \psi_N, \cN_{\bot \ph_\text{GP} } \psi_N \rangle$ and $\sigma^2 = \lim_{N \rightarrow \infty} \langle \psi_N, \cN_{\bot \ph_\text{GP} }^2 \psi_N \rangle $ denote the asymptotic mean and, respectively, variance of $\cN_{\bot \ph_\text{GP} }$. In other words, the characterization \eqref{eq:gf1} and the convexity of the rate function $\Lambda$ imply a large deviation bound similar to Hoeffding's inequality from probability. Let us also remark that \eqref{eq:ldp2} recovers similar large deviation estimates derived earlier in \cite{NR}. We expect that bounds analogous to \eqref{eq:ldp1} hold true in the trapped setting as well, but we leave this open for possible future work. 
\end{remark}



\subsection{Structure of the paper}
The rest of this paper is structured as follows. In Section \ref{sec:outline} we illustrate the main strategy of our proof of Theorem \ref{thm:main} in the simplified mean-field scaling. In Section \ref{sec:exc} we explain how to extend these ideas from the mean-field regime to the more challenging Gross-Pitaevskii regime and we prove Theorem \ref{thm:main} based on our technical key result. The proof of this latter result is carried out in the last Section \ref{sec:GN}. At several steps, we make use of various properties of the Gross-Pitaevskii functional and its minimizer which we recall for completeness in Appendix \ref{app:GP}.

\section{Exponential Moment Bound for Mean Field Bosons} \label{sec:outline}
In this section, we prove an analogue of Theorem \ref{thm:main} in the simplified setting of translation invariant mean field systems. In this setting, it is straightforward to obtain exponential control of the orthogonal excitations around the condensate based on a simple commutator argument. Although the proof of Theorem \ref{thm:main} is technically substantially more difficult than the proof of Theorem \ref{thm:mf} below, the overall strategy is conceptually similar. We therefore include the present section for illustrative purposes. 

Consider $N$ bosons moving in $\bT^3$ whose states are described by normalized wave functions in $L_s^2(\bT^{3N})$ and whose energies are described by the Hamilton operator
		\be \label{def:HNmf}  H_N^{\text {mf}} =  \sum_{i=1}^N -\Delta_{x_i} + \frac1N \sum_{1\leq i<j\leq N} v(x_i-x_j).  \ee
Following \cite{Sei}, we assume that $ v\in C(\bT^3)$ has a positive Fourier transform, that is $\widehat v(p) = \int_{\bT^3} dx\, e^{-ipx} v(x)\geq 0$ for every $p\in\Lambda^* = 2\pi\bZ^3$. In terms of the standard bosonic creation and annihilation operators $ a^*_p, a_p$, for $p \in\Lambda^*$, that create and, respectively, annihilate a plane wave $ x\mapsto \ph_p(x) = e^{ipx} \in L^2(\bT^3)$ of momentum $p\in\Lambda^*$ (see \eqref{eq:ccr} in Section \ref{sec:exc} for more details) , we have that
		\[ H_N^{\text {mf}} = \sum_{p\in\Lambda^*} |p|^2 a^*_p a_p + \frac1N \sum_{p,q,r\in\Lambda^*} \wh v(r) a^*_{p+r} a^*_{q-r} a_pa_q.\] 
As explained in detail in \cite{Sei}, $H_N^{\text {mf}}$ is bounded from below by 
		\be\label{eq:coercmf} H_N^{\text {mf}} \geq \frac{N}2 \widehat v(0) + \sum_{p\in\Lambda^*} |p|^2 a^*_p a_p - C \geq  \frac{N}2 \widehat v(0) + \cN_+ - C,  \ee
where we denote by $ \cN_+  = \sum_{p\in\Lambda^*_+}   a^*_p a_p$, $\Lambda^*_+= \Lambda^*\setminus\{0\} $, the number of orthogonal excitations around the constant condensate wave function $\ph_0 = 1_{|\bT^3} \in L^2(\bT^3)$. Testing $ H_N^{\text {mf}}$ with regards to $ \ph_0^{\otimes N}$, \eqref{eq:coercmf} implies that $ E_N^{\text{mf}} = \inf \sigma\big( H_N^{\text {mf}}\big)$ is equal to $ E_N^{\text{mf}} = \frac{N}2 \widehat v(0) + O(1)$ and that every normalized, approximate ground state $ \psi_N$ that satisfies 
		\[ \langle \psi_N,  H_N^{\text {mf}} \psi_N\rangle \leq E_N^{\text{mf}} + \zeta,\]
exhibits complete BEC into $\ph_0$ with bounded number of excitations $\cN_+$, that is
		\[    \langle \psi_N, \cN_+\psi_N\rangle \leq   \langle \psi_N, ( H_N^{\text {mf}}- E_N^{\text {mf}} ) \psi_N\rangle+C  \leq  C+ \zeta.    \]
The following theorem improves the last bound by showing that not only the first, but also an exponential moment of $\cN_+$ is of order one, uniformly on spectral subspaces of low energy. This is an analogue of Theorem \ref{thm:main} in the simplified mean field setting. 
\begin{theorem} \label{thm:mf}
Let $\zeta>0$ and denote by $E_N^{\emph{mf}}$ the ground state energy of $H_N^{\emph{mf}}$, defined in \eqref{def:HNmf}. Then, there exists a constant $C>0$ such that for all normalized vectors $\psi_N = {\bf 1}_{ (-\infty, E_N^{\emph{mf}} + \zeta] } (H_N^{\emph{mf}}) \psi_N $, $\| \psi_N \| =1$ and every $\kappa >0$ sufficiently small, we have that
		\[ \langle \psi_N, e^{\kappa \mathcal{N}_+} \psi_N \rangle \leq C .  \]
\end{theorem}
\begin{proof}
Denote by 
		$$ \cQ_\zeta = {\bf 1}_{ (-\infty, E_N^{\text{mf}} + \zeta] } (H_N^{\text{mf}})\big( L^2_s(\bT^{3N}) \big) \cap \big\{\psi_N\in L^2_s(\bT^{3N}): \|\psi_N \|=1 \big\} $$ 
the low energy spectral subspace of $ H_N^{\text{mf}}$ up to energies $E_N^{\text{mf}} + \zeta $ intersected with the unit sphere in $L^2_s(\bT^{3N}) $. We first claim that for sufficiently small $\kappa>0$, we have that
		\be \label{eq:cntrct1} \sup_{ \psi_N \in \cQ_\zeta }  \langle \psi_N, \cN_+ e^{ \kappa\mathcal{N}_+}  \psi_N\rangle \leq \frac{C +\zeta^{\frac12}}{1-C\kappa}  \sup_{ \psi_N \in \cQ_\zeta }  \langle \psi_N, e^{\kappa \mathcal{N}_+}   \psi_N\rangle  \ee
for some constant $C>0$ that is independent of $N$ and $\kappa$. To prove \eqref{eq:cntrct1}, we bound
		\[\begin{split}
		 \sup_{ \psi_N \in \cQ_\zeta }  \langle \psi_N, \cN_+  e^{ \kappa \mathcal{N}_+} \psi_N\rangle & =  \sup_{ \psi_N \in \cQ_\zeta }  \langle \psi_N, e^{ \frac \kappa2 \mathcal{N}_+} \cN_+ e^{ \frac \kappa2 \mathcal{N}_+} \psi_N\rangle \\
		 &\leq  \sup_{ \psi_N \in \cQ_\zeta }  \langle \psi_N, e^{ \frac \kappa2 \mathcal{N}_+}(H_N^{\text{mf}}- E_N^{\text{mf}} + C   ) e^{ \frac \kappa2 \mathcal{N}_+} \psi_N\rangle  \\
		 &\leq \zeta^{\frac12} \sup_{ \psi_N \in \cQ_\zeta } \Big|\Big\langle e^{ \frac \kappa2 \mathcal{N}_+}   \psi_N,  e^{ \frac \kappa2 \mathcal{N}_+}   \frac{ (H_N^{\text{mf}}- E_N^{\text{mf}}   ) \psi_N}{\| (H_N^{\text{mf}}- E_N^{\text{mf}}   ) \psi_N\| } \Big\rangle \Big|  \\
		 &\hspace{0.35cm} + \sup_{ \psi_N \in \cQ_\zeta }  \langle \psi_N, e^{ \frac \kappa2 \mathcal{N}_+}  [H_N^{\text{mf}},  e^{ \frac \kappa2 \mathcal{N}_+}  ] \psi_N\rangle+C\!\!\sup_{ \psi_N \in \cQ_\zeta }  \langle \psi_N, e^{  \kappa \mathcal{N}_+}  \psi_N\rangle,\\
		 \end{split}\]
where we used \eqref{eq:coercmf}. Since $  \frac{ (H_N^{\text{mf}}- E_N^{\text{mf}}   ) \psi_N }{\| (H_N^{\text{mf}}- E_N^{\text{mf}}   ) \psi_N\| } \in \cQ_\zeta $ whenever $\psi_N\in \cQ_\zeta$, we get
		\[\begin{split}
		 &  \sup_{ \psi_N \in \cQ_\zeta }  \langle \psi_N, \cN_+ e^{ \kappa \mathcal{N}_+}  \psi_N\rangle \\
		 & \leq (C +\zeta^{\frac12})  \sup_{ \psi_N \in \cQ_\zeta }  \langle \psi_N, e^{  \kappa \mathcal{N}_+}  \psi_N\rangle+  \sup_{ \psi_N \in \cQ_\zeta }  \langle \psi_N, e^{ \frac \kappa2 \mathcal{N}_+} [H_N^{\text{mf}},  e^{ \frac \kappa2 \mathcal{N}_+}] \psi_N\rangle.
		 \end{split}\]		
The commutator on the right hand side is readily found to be 
		\[\begin{split}
		[H_N^{\text{mf}},  e^{ \frac \kappa2 \mathcal{N}_+}] & = \frac12 \sum_{r\in\Lambda_+^*} \wh v(r) a^*_ra^*_{-r} \frac{a_0}{\sqrt N }   \frac{a_0}{\sqrt N }  e^{ \frac \kappa2 \mathcal{N}_+}  ( 1- e^{ \kappa } )  \\
		&\hspace{0.4cm}+ \frac1{\sqrt N} \sum_{p, r \in\Lambda_+^*:p+r\neq 0} \wh v(r) a^*_{p+r}  a^*_{-r} a_{p}  \frac{a_0}{\sqrt N }\, e^{ \frac \kappa2 \mathcal{N}_+}  ( 1- e^{ \frac \kappa2 })  - \text{h.c.},
 		\end{split} \]
so that standard bounds on the creation and annihilation operators together with a first order Taylor expansion of the differences $e^{ \frac {j \kappa} 2  }-1 $, for $j\in \{1,2\}$, imply that 
		\[  |\langle \psi_N,  e^{ \frac \kappa2 \mathcal{N}_+}   [H_N^{\text{mf}},  e^{ \frac \kappa2 \mathcal{N}_+}  ]  \psi_N  \rangle   | \leq  C  \kappa  \langle \psi_N, \cN_+ e^{ \kappa \mathcal{N}_+} \psi_N\rangle    \]  	
for some $C>0$, independent of $N, \kappa$ and $ \psi_N \in L^2_s(\bT^{3N})$. Combined with the previously derived bounds we thus get for sufficiently small $\kappa>0$ that 
		\[   \sup_{ \psi_N \in \cQ_\zeta }  \langle \psi_N, \cN_+e^{  \kappa\mathcal{N}_+}  \psi_N\rangle \leq  \frac{C +\zeta^{\frac12}}{1-C\kappa}  \sup_{ \psi_N \in \cQ_\zeta }  \langle \psi_N, e^{ \kappa \mathcal{N}_+}  \psi_N\rangle,  \]
which proves the bound \eqref{eq:cntrct1}. 

To conclude the theorem, observe finally that \eqref{eq:cntrct1} implies
		\[\begin{split}
		  \sup_{ \psi_N \in \cQ_\zeta }  \langle \psi_N, e^{\kappa\mathcal{N}_+}   \psi_N\rangle &=1 +  \sup_{ \psi_N \in \cQ_\zeta } \int_0^\kappa dt\,   \langle \psi_N, \cN_+ e^{t \mathcal{N}_+}    \psi_N\rangle\\
		& \leq 1 +   \kappa \sup_{t \in [0, \kappa]} \sup_{ \psi_N \in \cQ_\zeta }   \langle \psi_N, \cN_+  e^{t \mathcal{N}_+}   \psi_N\rangle \\
		&\leq  1 +    \frac{ C\kappa +\zeta^{\frac12} \kappa}{1-C\kappa}   \sup_{ \psi_N \in \cQ_\zeta }   \langle \psi_N, e^{\kappa\mathcal{N}_+}   \psi_N\rangle
		\end{split}\]
so that for sufficiently small $\kappa>0$, we conclude
		\[  \sup_{ \psi_N \in \cQ_\zeta }  \langle \psi_N, e^{\kappa \mathcal{N}_+}   \psi_N\rangle \leq \frac{1}{1- C' } \hspace{0.5cm} \text{for} \hspace{0.5cm} C' = 1-  \frac{ C\kappa +\zeta^{\frac12}\kappa}{1-C\kappa} <1.  \]	
\end{proof}

\section{Renormalized Hamiltonian and Proof of Theorem \ref{thm:main}} \label{sec:exc}

In this section, we prove Theorem \ref{thm:main}. As mentioned earlier, this is based on a similar argument as in the proof of Theorem \ref{thm:mf}. While an analogue of \eqref{eq:coercmf} is given by \eqref{eq:coerc}, previously derived in \cite{BSS1}, the key difficulty is that we can not apply the commutator strategy directly to the shifted GP Hamiltonian $H_N - N \cE_{\text{GP}}(\ph_{\text{GP}})$, for $H_N$ defined in \eqref{def:ham}. Instead, we first need to renormalize $H_N$ through a quadratic generalized Bogoliubov transformation that takes into account short-scale correlations among the particles. This regularizes certain energy contributions to $H_N$ which enables us to apply a commutator argument as in Section \ref{sec:outline} on the level of the renormalized Hamiltonian.  

In order to describe the key steps in detail, we first introduce some notation and we recall several results previously established in \cite{BSS1}. For our analysis, we find it convenient to switch to a Fock space setting. Recall that the bosonic Fock space is defined as
		\[ \cF  = \mathbb C \oplus  \bigoplus_{n \geq 1} L^2 (\bR^3)^{\otimes_{\text{sym}} n}. \]
The standard creation and annihilation operators on $\cF$ are defined by
		\[ \begin{split} 
		(a^* (f) \Psi)^{(n)} (x_1, \dots , x_n) &= \frac{1}{\sqrt{n}} \sum_{j=1}^n f (x_j) \Psi^{(n-1)} (x_1, \dots , x_{j-1}, x_{j+1} , \dots , x_n), \\
		(a (g) \Psi)^{(n)} (x_1, \dots , x_n) &= \sqrt{n+1} \int \bar{g} (x) \Psi^{(n+1)} (x,x_1, \dots , x_n) 
		\end{split} \]
and they satisfy, for every $f,g \in L^2 (\bR^3)$, the canonical commutation relations 
		\be \label{eq:ccr} [a (f), a^* (g) ] = \langle f,g \rangle , \quad [ a(f), a(g)] = [a^* (f), a^* (g) ] = 0. \ee
We also introduce operator-valued distributions $a_x, a_x^*$, for $x \in \bR^3$, through
		\[ a(f) = \int dx\, \overline{f} (x) \, a_x  , \quad a^* (g) = \int dx\, g(x) \, a_x^*  .  \]
To switch from $L^2_s(\bR^{3N}) $ to the Fock space, we follow \cite{LNSS} and use the decomposition  
		\[ \psi_N = \xi_0 \ph_\text{GP}^{\otimes N} + \xi_1 \otimes_s \ph_\text{GP}^{\otimes (N-1)} + \ldots + \xi_{N-1}\otimes_s \ph_\text{GP}+ \xi_N, \]
where $\xi_j \in L^2_{\perp \ph_\text{GP}} (\bR^3)^{\otimes_\text{sym} j}$, for every $\psi_N\in L^2_s(\bR^{3N})$. This defines a unitary map 
		$$L^2_s(\bR^{3N})\ni \psi_N \mapsto  U_N \psi_N = ( \xi_0, \xi_1, \dots , \xi_N ) \in \cF_{\bot\ph_\text{GP}}^{\leq N} =\bigoplus_{n = 0}^N L_{\perp \ph_\text{GP}}^2 (\bR^3)^{\otimes_s n} $$ 
that sends $\psi_N$ to the excitation vector $ ( \xi_0, \xi_1, \dots , \xi_N )$, which describes the fluctuations of $\psi_N$ around the pure condensate $\ph_{\text{GP}}^{\otimes N}$. Note that $  \cF_{\bot\ph_\text{GP}}^{\leq N}$ is a truncated Fock space built over $L^2_{\perp \ph_\text{GP}} (\bR^3)$, the orthogonal complement of the condensate state $\ph_\text{GP} \in L^2 (\bR^3)$, the unique normalized, positive minimizer of $\cE_\text{GP}$, defined in \eqref{def:EGP}. 

The action of $U_N$ on the creation and annihilation operators is explicit and given by
		\begin{equation}\label{2.1.UNconjugation}
 		\begin{split}
    		U_{N} a^* (\ph_\text{GP} ) a (\pn) U_{N}^*  & = N-\cN, \\
    		U_{N} a^*(f)a (\pn) U_{N}^* &= a^*(f) \sqrt{N-\cN}  , \\
    		U_{N} a^* (\pn) a(g) U_{N}^* &= \sqrt{N-\cN} a(g)  , \\
    		U_{N} a^* (f) a(g)U_{N}^* &= a^*(f)a(g)
    		\end{split}
    		\end{equation}
for every $ f,g\in L^2_{\perp \pn} (\bR^3)$. Here, $ \cN$ denotes the number of particles operator in the excitation Fock space $  \cF_{\bot\ph_\text{GP}}^{\leq N}$, which is the multiplication operator given by
		\[\cN ( \xi_j)_{j=0}^N = ( j \xi_j)_{j=0}^N\in \cF_{\bot\ph_\text{GP}}^{\leq N}.  \]
Notice in particular that the rules \eqref{2.1.UNconjugation} imply that $ U_N \cN_{\bot \ph_\text{GP} } U_N^* = \cN$. Moreover, \eqref{2.1.UNconjugation} implies that the transformed Hamiltonian $\cL_N = U_N H_N U_N^*$ can be decomposed into
		\[ \cL_N =  \cL^{(0)}_{N}+\cL^{(1)}_{N} + \cL^{(2)}_{N} + \cL^{(3)}_{N} + \cL^{(4)}_{N} , \]
where, in the sense of quadratic forms on $  \cF_{\bot\ph_\text{GP}}^{\leq N}\times \cF_{\bot\ph_\text{GP}}^{\leq N}$, we have that
	\begin{equation}\label{eq:cLNj}
	\begin{split}
	\mathcal{L}_N^{(0)} &= \big\langle \pn, \big[ -\Delta + V_\text{ext} + \frac{1}{2} \big(N^3 V(N\cdot) * \vert \pn \vert^2\big) \big] \pn \big\rangle (N-\cN) \\
	& \qquad -  \frac12\big\langle \pn,  \big(N^3 V(N\cdot) * \vert \pn \vert^2 \big)\pn \big\rangle (\cN+1)(1-\cN/N) ,  \\
	\mathcal{L}_N^{(1)} &=   \sqrt{N} b\left(  \left( N^3 V(N\cdot) * \vert \pn\vert^2   -8\pi \frak{a}_0|\pn|^2 \right)\pn\right)\\
	&\hspace{0.5cm} - \frac{\mathcal{N}+1}{\sqrt N}  b\left( \left( N^3 V(N\cdot) * \vert \pn\vert^2 \right) \pn \right) + \text{h.c.},\\
	\mathcal{L}_N^{(2)} &= \int dx\; \left( a_x^* (-\Delta_x ) a_x + V_\text{ext}(x)a_x^{*} a_x \right)
		\\
	&\hspace{0.5cm} + \int  dxdy\;  N^3 V(N(x-y)) \vert\pn(y) \vert^2 \Big(b_x^* b_x - \frac{1}{N} a_x^* a_x \Big)   \\
	&\hspace{0.5cm}+ \int   dxdy\; N^3 V(N(x-y)) \pn(x) \pn(y) \Big( b_x^* b_y - \frac{1}{N} a_x^* a_y \Big)  \\
	&\hspace{0.5cm}+ \frac{1}{2}  \int  dxdy\;  N^3 V(N(x-y)) \pn(y) \pn(x) \Big(b_x^* b_y^*  + \text{h.c.} \Big), \\
	\mathcal{L}_N^{(3)} &=  \int  dx dy\, N^{ 5/2} V(N(x-y)) \pn(y) \big( b_x^* a^*_y a_x + \text{h.c.} \big), \\
	\mathcal{L}_N^{(4)} &= \frac{1}{2} \int  dxdy\; N^2 V(N(x-y)) a_x^* a_y^* a_y a_x .
	\end{split}
	\end{equation}	
Here, we introduced the modified creation and annihilation operators
		\[ b^* (f) = \int dx\, f(x) b^*_x = a^* (f) \, \sqrt{\frac{N-\cN}{N}} , \quad \quad b (f) =  \int dx\, \overline f(x) b_x =  \sqrt{\frac{N- \cN}{N}} \, a (f) . \]
 
We renormalize $\cL_N$ through a generalized Bogoliubov transformation, which we implement as in \cite{BSS1}. To define this precisely, consider the Neumann ground state of
		\[ \Big( -\Delta + \frac{1}{2} V \Big) f_{\ell} = \lambda_{\ell} f_\ell \]
in $B_{N\ell}(0)\subset \bR^3$, for some $0 < \ell < 1$ which we assume to be sufficiently small (but fixed, independently of $N$). By radial symmetry of $V$, $f_\ell$ is radial and we normalize it such that $f_\ell (x) = 1$ if $|x| = N \ell$. The rescaled solution $f_\ell (N.)$ solves 
		\[ \Big( -\Delta + \frac{ N^2}{2} V (N.) \Big) f_\ell (N.) = N^2 \lambda_\ell f_\ell (N.) \]
in $B_{\ell}(0)$. In the following, we identify $f_\ell (N.)$ with its extension to $\bR^3$, defined by $f_{N,\ell} (x) = 1$ whenever $|x| > \ell$. 
Moreover, we set $ w_\ell = 1-f_\ell $ so that $ w_\ell(N.)$ has compact support in $ B_\ell(0)$. We denote the Fourier transform of $ w_\ell $ by
		\[ \widehat w_\ell (p) = \int dx\; w_\ell(x)e^{-2\pi ipx}.  \]
From \cite[Lemma 3.1]{BSS1}, we recall that for $x,p\in\bR^3$, it holds true that
		\[ 
		w_\ell(x)\leq \frac{C}{|x|+1}, \;\; |\nabla w_\ell(x)|\leq \frac{C }{|x|^2+1}, \;\; |\widehat{w}_\ell (p)| \leq \frac{C}{|p|^2}. 
		\]
Finally, denoting by $\chi_H$ the characteristic function of $\{ p\in\bR^3: |p|\geq \ell^{-\alpha}\} \subset\bR^3$ and by $ \check{\chi}_H \in L^2 (\bR^3)$ its inverse Fourier transform, we define the correlation kernel
		\be \label{eq:defeta}  \eta =  Q\otimes Q  \big( -(Nw_\ell(N.)  \ast \wch \chi_H ) (x-y)  \pn\otimes \pn\big) \in L^2_{\bot \pn}(\bR^3) \otimes_\text{sym}L^2_{\bot \pn}(\bR^3)   \ee
We remark that the only difference of $\eta $ compared to $\eta_H$, defined in \cite[Eq. (3.11)]{BSS1}, consists of the prefactor $ Q\otimes Q$, for $Q = 1-|\pn\rangle\langle\pn|$. This difference is not substantial and proceeding as in \cite[Lemma 3.2]{BSS1}, we readily obtain the following Lemma. 
\begin{lemma}\label{lm:bndseta}
Let $V, V_\emph{ext}$ satisfy Assumption \ref{ass:V}, let $\ell\in (0,1)$ and let $\alpha>0$. Set $\eta_{x}(y)= \eta (x,y)$ for $x\in\bR^3$. Then, there exists a constant $C>0$, that is independent of $N$, $\ell\in (0,1)$ and $x\in\bR^3$, such that 
	\[
	\Vert \eta  \Vert \leq C \ell^{\frac\alpha 2}, \hspace{0.5cm}  \Vert \eta_{x} \Vert \leq C \ell^{\frac\alpha2} \vert \pn(x) \vert, \hspace{0.5cm} \Vert \nabla_1 \eta  \Vert, \,\Vert \nabla_2 \eta  \Vert  \leq C \sqrt{N}.
	\]
	Furthermore, identifying $\eta (x,y)$ as the kernel of a Hilbert-Schmidt operator and $\eta ^{(n)} (x,y)$ with the kernel of its $n$-th power, we have for $n\geq 2$ and $x,y\in\bR^3$ that 
	\[
	\vert \eta ^{(n)} (x,y) \vert 
	\leq \Vert \eta_{x} \Vert  \Vert \eta_{y} \Vert  \Vert \eta  \Vert^{n-2} \leq  C \ell^{\alpha} \Vert \eta  \Vert^{n-2} \vert \pn(x) \vert   \vert \pn(y) \vert
	\]
and, for all $N$ sufficiently large, that
	\[ |\eta (x,y)|\leq CN |\pn(x) | |\pn(y)| \leq CN. \]
\end{lemma}
 
With the kernel $\eta $ from \eqref{eq:defeta}, define the generalized Bogoliubov transformation
		\be \begin{split}\label{eq:eB} 
		e^{ B(\eta)}\hspace{0.5cm} \text{for}\hspace{0.5cm}  B(\eta) = \frac{1}{2} \int dx dy \, \eta  (x,y) \,  \big(  b_x^* b_y^* - b_x b_y \big).
		\end{split} \ee
Then $e^{ B(\eta)}:  \cF_{\bot\ph_\text{GP}}^{\leq N}\to  \cF_{\bot\ph_\text{GP}}^{\leq N}$ is a unitary map. Following \cite{BSS1}, we use this map to renormalize the excitation Hamiltonian $\cL_N$ and define
		\be \label{def:GN} \cG_{N} = e^{-B(\eta)} \cL_N e^{B(\eta)}.  \ee
The next result collects the key properties of $\cG_N$ needed for our proof of Theorem \ref{thm:main}. 		
\begin{prop} \label {prop:GN}
Let $V, V_\emph{ext}$ satisfy Assumption \ref{ass:V}, let $\eta$ be as in \eqref{eq:defeta} (for $\ell >0$ small enough) and let $\cG_{N}$ be as in \eqref{def:GN}. Moreover, set $V_\emph{ext}' =  V_\emph{ext} +\Lambda $ for some $\Lambda\geq 0$ such that $V_\emph{ext}'\geq 0$ in $\bR^3$ and set 
		\[ \cH_N = \int dx\, a^*_x \big(-\Delta_x + V_\emph{ext}'(x)\big) a_x  + \frac12 \int dxdy \, N^2 V(N(x-y))a^*_xa^*_y   a_xa_y. \]
Then, there exist $c, C>0$, independent of $N$ and $\ell$, such that 
		\be  \begin{split} \label{eq:GNbnd1}
		 \cG_N \geq N \cE_\emph{GP}(\ph_\emph{GP}) + c \, (\cH_N + \cN) - C. 
		\end{split}\ee
Moreover, for every $\kappa >0$ sufficiently small, $\cG_N$ satisfies the commutator bound
		\be  \begin{split} \label{eq:GNbnd2}
		 |  \langle \xi,  e^{\kappa \cN/2} [\cG_N, e^{\kappa \cN/2} ] \xi\rangle | \leq  C \kappa \, \langle \xi,(\cH_N + \cN)e^{\kappa \cN}\xi\rangle
		\end{split}\ee
for some constant $C>0$ that is uniform in $N ,  \kappa$ and $\xi \in  \cF_{\bot\ph_\emph{GP}}^{\leq N}$.
\end{prop}
Let us point out that the bound \eqref{eq:GNbnd1} is essentially a direct consequence of \cite[Prop. 3.4]{BSS1} and the lower bound \eqref{eq:coerc}, proved in \cite[Theorem 1.1]{BSS1}. A commutator estimate as in \eqref{eq:GNbnd2} is not proved in \cite{BSS1}, but follows with very similar arguments as presented in great detail in \cite{BS, BBCS1, BBCS2, BBCS3, BBCS4, BSS1, BSS2}. Due its close similarity to the analysis in \cite{BSS1}, we focus in the proof of Prop.\ \ref{prop:GN} on the key steps only and we defer the proof to Section \ref{sec:GN}.

In addition to Prop.\ \ref{prop:GN}, we need the following Lemma. 
\begin{lemma}\label{lm:conj}
Let $V, V_\emph{ext}$ satisfy Assumption \ref{ass:V} and let $\eta$ be defined as in \eqref{eq:defeta}, for $\ell >0$ sufficiently small. Then, for every $\kappa >0$ small enough, we have that 
		\[    e^{B(\eta)} e^{\kappa \mathcal{N}} e^{-B(\eta)} \leq C   e^{ 2\kappa \mathcal{N}}. \]
\end{lemma}  
\begin{proof} The proof is a straightforward adaption of the proof of \cite[Lemma 3.1]{R} to the present setting. For completeness, we recall the main steps. Let $\xi \in  \cF_{\bot\ph_\text{GP}}^{\leq N}$ and define 
		\[ [0,1]  \ni s\mapsto  \| \xi_s\|^2 \text{ for } \xi_s = e^{ \kappa _s \cN} e^{-B(\eta)}\xi, \; \kappa _s = \frac12 (s+2(1-s) )\kappa . \]
Then a straightforward calculation (see also \eqref{eq:bounds-b2} for more details) shows that 
		\[\begin{split}
		\frac{d}{ds} \| \xi_s\|^2 &= -\kappa  \langle  \xi_s, \cN \xi_s  \rangle -4 \sinh\big( (2 -s)\kappa  \big) \text{Re} \int dx dy\, \eta(x,y) \langle \xi_s, b^*_x b^*_y\xi_s\rangle\\
		&\leq -\kappa  \langle  \xi_s, \cN \xi_s  2\rangle + 8 \|\eta\| \kappa  \langle  \xi_s, \cN \xi_s  \rangle + C \| \xi_s\|^2 < C \| \xi_s\|^2
		\end{split}\]		
for some $C>0$ independent of all parameters. Note that we used that $\|\eta\|$ is suitably small if $\ell$ is sufficiently small. Gronwall's lemma now implies the claim.
\end{proof}
 
We are now ready to prove Theorem \ref{thm:main}.
\begin{proof}[Proof of Theorem \ref{thm:main}] By Lemma \ref{lm:conj} and the unitary equivalence of $H_N$ with the renormalized excitation Hamiltonian $\cG_N = e^{-B(\eta)}U_N H_N U_N^* e^{B(\eta)}$, it suffices to show that
		\[   \sup_{ \xi \in \cS_\zeta }  \langle \xi, e^{\kappa  \mathcal{N} }   \xi\rangle \leq C, \;\text{ for }\; \;\cS_\zeta = {\bf 1}_{ (-\infty, E_N  + \zeta] } (\cG_N )\big( \cF_{\bot\ph_\text{GP}}^{\leq N}  \big) \cap \big\{\xi\in \cF_{\bot\ph_\text{GP}}^{\leq N}: \|\xi \|=1 \big\}.\]
Proceeding as in the proof of Theorem \ref{thm:mf}, we first claim that for small enough $\kappa$
		\be \label{eq:GPcntrct1}  \sup_{ \xi \in \cS_\zeta }  \langle \xi_N, (\cH_N+\cN) e^{ \kappa\mathcal{N} }  \xi_N\rangle \leq \frac{C +C\zeta^{\frac12}}{1-C\kappa}  \sup_{ \xi \in \cS_\zeta }  \langle \xi, e^{\kappa \mathcal{N}}   \xi\rangle,  \ee
for a constant $C>0$ that is independent of $N$ and $\kappa$. Observe here that $ [\cH_N, \cN]=0$. To prove the claim, notice indeed that the bounds from Proposition \ref{prop:GN} imply that
		\[\begin{split}
		 \sup_{ \xi \in \cS_\zeta }  \langle \xi, (\cH_N+\cN)   e^{ \kappa \mathcal{N}} \xi\rangle & =  \sup_{ \xi \in \cS_\zeta }  \langle \xi, e^{ \frac  \kappa2 \mathcal{N}}  (\cH_N+\cN)  e^{ \frac  \kappa2 \mathcal{N}} \xi\rangle \\
		 &\leq C \sup_{ \xi \in \cS_\zeta }  \langle \xi, e^{ \frac  \kappa2 \cN }(\cG_N - E_N + C  ) e^{ \frac  \kappa2 \cN } \xi\rangle  \\
		 &\leq C \zeta^{\frac12} \sup_{ \xi \in \cS_\zeta } \bigg|\bigg\langle e^{ \frac  \kappa2 \cN }   \xi,  e^{ \frac  \kappa2 \cN }   \frac{ (\cG_N-E_N    ) \xi}{\| (\cG_N-E_N    ) \xi \| } \bigg\rangle \bigg|  \\
		 &\hspace{0.35cm} + C\sup_{ \xi \in \cS_\zeta }  \langle \xi, e^{ \frac  \kappa2 \cN }  [\cG_N,  e^{ \frac  \kappa2 \cN }  ] \xi\rangle+C \sup_{ \xi \in \cS_\zeta }  \langle \xi, e^{   \kappa \cN }  \xi\rangle,\\
		  &\leq (C+C \zeta^{\frac12}) \sup_{ \xi \in \cS_\zeta }  \langle \xi, e^{   \kappa \cN }  \xi\rangle  + C \kappa \sup_{ \xi \in \cS_\zeta }  \langle \xi,  (\cH_N+\cN)   e^{  \kappa \mathcal{N}} \xi\rangle, 
		 \end{split}\]
which implies \eqref{eq:GPcntrct1}. Here, we used in the first inequality that $E_N = N\cE_\text{GP}(\ph_\text{GP})+ O(1)$, which is a direct consequence of the main results of \cite{NNRT, BSS1, NT, BSS2}. 	 

To conclude the theorem, note that \eqref{eq:GPcntrct1} and the positivity of $\cH_N$ imply that
		\[\begin{split}
		  \sup_{ \xi \in \cS_\zeta }  \langle \xi, e^{\kappa \cN }   \xi\rangle=1 +    \sup_{ \xi \in \cS_\zeta } \int_0^\kappa dt\,   \langle \xi, \cN  e^{t \cN }    \xi\rangle& \leq 1 +   \kappa   \sup_{ \xi \in \cS_\zeta }   \langle \xi, (\cH_N + \cN )  e^{\kappa \cN }   \xi\rangle \\
		&\leq  1 +    \frac{ C\kappa +C\zeta^{\frac12} \kappa}{1-C\kappa}   \sup_{ \xi \in \cS_\zeta }   \langle \xi, e^{\kappa \cN }   \xi\rangle.
		\end{split}\]
Choosing $\kappa>0$ sufficiently small, this implies
		\[  \sup_{ \xi \in \cS_\zeta }  \langle \xi, e^{\kappa \cN }   \xi\rangle \leq \frac{1}{1- C' } \hspace{0.5cm} \text{for} \hspace{0.5cm} C' = 1-  \frac{ C\kappa +C\zeta^{\frac12}\kappa}{1-C\kappa} <1.  \]	
\end{proof}

\section{Analysis of $\mathcal{G}_N$}\label{sec:GN}

To analyze the renormalized excitation Hamiltonian $\mathcal{G}_N = e^{-B(\eta)} \mathcal{L}_N e^{B(\eta)}$ as defined in \eqref{def:GN}, we recall from \eqref{eq:cLNj} the splitting of the excitation Hamiltonian $\mathcal{L}_N = U_NH_NU_N^*$  
\[
 \mathcal{L}_N = \mathcal{L}_N^{(0)} + \mathcal{L}_N^{(1)}+\mathcal{L}_N^{(2)}+\mathcal{L}_N^{(3)}+\mathcal{L}_N^{(4)}, 
\]
where, in the sense of quadratic forms on $\mathcal{F}_{\perp\varphi_{\text{GP}}}^{\leq N}$ the single contributions $\mathcal{L}_N^{(j)}$ are ordered w.r.t. to the number of standard resp. modified creation and annihilation operators. 
Based on this splitting, we then write the renormalized excitation Hamiltoninan as a sum of the following five terms 
\[
 \mathcal{G}_N = \sum\limits_{k=0}^4 \mathcal{G}_N^{(k)}\quad\text{with }\quad \mathcal{G}_N^{(k)}= e^{-B(\eta)}\mathcal{L}_N^{(k)}e^{B(\eta)}
\]
and analyze the single terms $\mathcal{G}_N^{(0)},\ldots,\mathcal{G}_N^{(4)}$ in Sections~\ref{sec:GN0}-\ref{sec:GN3} separately. 


Our analysis builds on properties of the generalized Bogoliubov transform, first proved in \cite{BS} and later studied in greater detail for example in \cite{BBCS2, BBCS3, BBCS4, BSS1,BSS2}. In the next Section \ref{sec:bogo-mod} we first briefly summarize useful properties, and then later analyze the operators $\mathcal{G}_N^{(k)}$ in Sections \ref{sec:GN0} to \ref{sec:GN3}.

\subsection{Modified Bogoliubov transformation}
\label{sec:bogo-mod}

The generalized Bogoliubov transformation's action on modified creation and annihilation operators is for any $f \in L^2( \mathbb{R}^3)$ given by 
\begin{align}
\label{eq:action-bogo-mod}
e^{- B( \eta)} b (f) e^{B ( \eta)} = b( \cosh_\eta(f) ) + b^*( \sinh_\eta( f)) + d_{\eta} (f) 
\end{align}
where the first two, leading order, terms are formulated w.r.t. the convergent series  
\[
\cosh_\eta (f) : = \sum_{n =0}^\infty \frac{\eta^{(2n)} (f)}{(2n)!}, \quad \text{and} \quad \sinh_\eta (f) : = \sum_{n =0}^\infty \frac{\eta^{(2n+1)} (f)}{(2n+1)!}, 
\]
whereas $d_{\eta } (f)$ is a sub-leading error of order $O( 1/N)$. 
For the (standard) Bogoliubov transformation, formulated w.r.t. to standard creation and annihilation operators $a^*(g),a(f)$ that satisfy the canonical commutation relations \eqref{eq:ccr}, the action on standard creation and annihilation operators of the form \eqref{eq:action-bogo-mod} is exact without any error. The generalized Bogoliubov transformation \eqref{eq:eB}, however, is formulated w.r.t. to modified creation and annihilation operators whose commutation relations 
\begin{align}
\label{eq:CCR-mod}
\big[ b(f), b^*(g) \big] = \langle f,g \rangle \bigg( 1- \frac{\mathcal{N}}{N} \bigg) - \frac{a^*(f) a(g)}{N}, \quad \text{and} \quad \big[ b^*(f), b^*(g) \big] = \big[ b(f), b(g) \big] =0 
\end{align}
come, compared with \eqref{eq:ccr}, with a correction of order $O(1/N)$ in the large particle limit yielding an sub-leading error in \eqref{eq:action-bogo-mod}. Similarly to the standard creation and annihilation operators, the modified ones are bounded w.r.t. the number of particles operators: For any $h \in L^2( \mathbb{R}^3)$ we have 
\[
\label{eq:bounds-b1}
\| b(h) \xi \| \leq \| f \| \; \| \mathcal{N}^{1/2} \xi \|, \quad \text{and} \quad \| b^*(h) \xi \| \leq \| f \| \; \| (\mathcal{N} + 1)^{1/2} \xi \| 
\]
for any $\xi \in \mathcal{F}_{\perp \varphi_{\rm GP}}^{\leq N}$. Moreover, 
for ${j\in L^2(\mathbb{R}^3\times\mathbb{R}^3)}$ and ${\sharp_\ell,\sharp_r\in \{\cdot,*\}}$ the operator 
\begin{align}
\label{eq:bounds-b2}
B_{\sharp_\ell,\sharp_r}(j) = \int dx\, b^{\sharp_\ell}(j_x)b_x^{\sharp_r} = \int dxdy\,j^{\overline{\sharp}_\ell}(x,y)b_y^{\sharp_\ell}b_x^{\sharp_r},
\end{align}
satisfies for any ${\xi \in \mathcal{F}^{\leq N}}$ 
\[
||B_{\sharp_\ell,\sharp_r}(j)\psi||\leq C  ||(\mathcal{N}+1)\psi| \begin{cases} ||j||_2+\int |j(x,x)|\,dx &\text{if }\sharp_\ell=\cdot,\sharp_r=*, \\ ||j||_2 & \text{otherwise,}\end{cases} 
\]
for some $C>0$ (see for example \cite[Lemma 2.1]{BS}). 

A key ingredient for our proof is the approximate action \eqref{eq:action-bogo-mod} of the the generalized Bogoliubov transform.  To control the error in \eqref{eq:action-bogo-mod}, we use ideas and concepts introduced first in \cite{BS} where it was observed that the action of the modified Bogoliubov transformation on $b^*(f), b(f)$ can be expressed as an infinite series of nested commutators 
\[
\text{ad}_{B(\eta)}^{(n)}(A) = [B(\eta),\text{ad}_{B(\eta)}^{(n-1)}(A)],\quad  \text{ad}_{B(\eta)}^{(0)}(A)=A \; . 
\]
for any operator $A$. More precisely, \cite[Lemma 3.3]{BS} proves that for any $\eta \in L^2( \mathbb{R}^3 \times L^2( \mathbb{R}^3)$ that is symmetric and has sufficiently small norm $\| \eta \| $, the action of the generalized Bogoliubov transform is given by the infinite sums 
\[
\label{eq:expand1}
e^{-B(\eta)}b(f)e^{B(\eta)} =\sum\limits_{n=0}^\infty \frac{(-1)^n}{n!} \text{ad}_{B(\eta)}^{(n)}(b(f))
\]
resp. 
\[
e^{-B(\eta)}b^*(f)e^{B(\eta)} = \sum\limits_{n=0}^\infty \frac{(-1)^n}{n!}\text{ad}_{B(\eta)}^{(n)}(b^*(f))
\]
and the series of the r.h.s.\ converge absolutely. Therefore, the analysis of the r.h.s.\ of \eqref{eq:action-bogo-mod}, and in particular of the error $d_\eta (f)$, builds on a detailed study of nested commutators $\text{ad}_{B( \eta)}^{(n)} (b^\sharp (f))$ with $\sharp = \cdot, * $ and is carried out in detail in \cite{BS, BBCS2, BBCS3, BBCS4, BSS1, BSS2}. In view of the following, see in particular \cite[Lemma 5.1]{BS} and \cite[Lemma 4.1]{BSS1}. In the next lemma, we summarize properties of the point-wise errors $d_{\eta,x}^*, d_{\eta,x}$ (interpreted in operator-valued distributional sense) defined through 
\begin{align}
e^{- B( \eta)} b_x^* e^{B( \eta)} =  b^*( \cosh_{\eta,x} ) + b( \sinh_{\eta,x} ) + d_{\eta, x}^* \label{def:dxstar}
\end{align}
resp.
\begin{align}
e^{- B( \eta)} b_x e^{B( \eta)} =  b( \cosh_{\eta,x} ) + b^*( \sinh_{\eta,x} ) + d_{\eta, x} \label{def:dx}
\end{align}
with the  notation $ \cosh_{\eta,x} (y) = \cosh_{\eta} (y,x)$ resp.  $ \sinh_{\eta,x} (y) = \sinh_{\eta} (y,x)$. 

\begin{lemma}
\label{lemma:error-bounds1}
Let $ x \in \mathbb{R}^3$, $n \in \mathbb{Z}$ and $d_{\eta,x}^*, d_{\eta,x}$ be defined by \eqref{def:dxstar} resp. \eqref{def:dx}. Then there exists a constant $C_n>0$ such that for every $\xi \in \mathcal{F}_{\perp \varphi_{\rm GP}}^{\leq N}$
\begin{align}
\| ( \mathcal{N} + 1)^{n/2} d_{\eta,x} \xi \| \leq& \frac{C_n}{N} \bigg[ \| b_x ( \mathcal{N} + 1)^{(n+2)/2} \xi \| + \| \eta_x \| \; \| ( \mathcal{N} + 1)^{(n+3)/2} \xi \| \bigg] 
\end{align}
Furthermore, 
\begin{align}
\| ( \mathcal{N} &  + 1)^{n/2} b_x d_{\eta, y}  \xi \| \notag \\
\leq& \frac{C_n }{N} \bigg[ \| \eta \|^2 \| ( \mathcal{N} + 1)^{(n+4)/2} \xi \| + \vert \eta (x;y) \vert \; \| \eta \| \; \| ( \mathcal{N} + 1)^{(n+4)/2}  \xi \| \notag \\
& \quad \quad \quad + \| \eta \| \; \|a_x ( \mathcal{N} + 1)^{(n+3)/2}  \xi \| + \| \eta \|^2 \| a_xa_y ( \mathcal{N} + 1)^{(n+2)/2}  \xi \| \bigg]  \; , 
\end{align}
and 
\begin{align}
\|   ( \mathcal{N}  & + 1)^{n/2} d_{ \eta, x}d_{ \eta, y} \xi \| \notag \\
\leq& \frac{C_n }{N^2} \bigg[ \| \eta \|^2 \| ( \mathcal{N} + 1)^{n+6)/2} e^{\lambda \mathcal{N}} \xi \| +  \vert  \eta(x;y) \vert \; \| \eta \| \| ( \mathcal{N} + 1)^{(n+4)/2}  \xi \| \notag \\
& \quad \quad \quad  + \| \eta \|^2 \| a_x ( \mathcal{N} +1)^{(n+5)/2} \xi \| + \| \eta \|^2 \| a_y \| ( \mathcal{N} + 1)^{(n+4)/2} \xi \| \notag \\
& \quad \quad \quad + \| \eta \|^2 \| a_xa_y( \mathcal{N} + 1)^{(n+4)/2}  \xi \| \bigg] \; . 
\end{align}
\end{lemma}

The proof of Proposition \ref{prop:GN} requires not only bounds of the errors $d_{\eta,x}, d_{\eta,x}^*$ as in the previous Lemma, but furthermore similar bounds for commutators of $d_{\eta,x}, d_{\eta,x}^*$ with $e^{\lambda \mathcal{N}}$. We collect these error bounds in the lemma below whose proof goes bac to {NR}.

\begin{lemma}
\label{lemma:error-bounds2} Let $x,y \in \mathbb{R}^3 $, $n \in \mathbb{Z}$  and $d_{\eta,x}, d_{\eta,x}^*$ be defined by \eqref{def:dxstar} resp. \eqref{def:dx}. For sufficiently small $\| \eta \|$ and $0<\lambda <1$, there exists $C>0$ such that 
\[
\| ( \mathcal{N} + 1)^{n/2} \big[ e^{\lambda \mathcal{N}}, d_{\eta,x} \big] \xi \| \leq \frac{C \lambda }{N} \bigg[ \| b_x ( \mathcal{N} + 1)^{(n+2)/2} e^{\lambda \mathcal{N}} \xi \| + \| \eta_x \| \; \| ( \mathcal{N}+ 1)^{(n+3)/2} e^{\lambda \mathcal{N}} \xi \| \bigg] \; , 
\]
and 
\begin{align}
\| ( \mathcal{N}  & + 1)^{n/2}   e^{-\lambda \mathcal{N}} \big[ e^{\lambda \mathcal{N}}, \big[ e^{\lambda \mathcal{N}}, d_{\eta,x} \big] \big]  \xi \|\notag \\
 \leq &  \frac{C \lambda^2 }{N} \bigg[ \| b_x ( \mathcal{N} + 1)^{(n+2)/2} e^{\lambda \mathcal{N}} \xi \| + \| \eta_x \| \; \| ( \mathcal{N}+ 1)^{(n+3)/2} e^{\lambda \mathcal{N}} \xi \| \bigg] \; . 
\end{align}
Furthermore, 
\begin{align}
\| ( \mathcal{N} &  + 1)^{n/2} \big[ e^{\lambda \mathcal{N}}, b_x d_y \big] \xi \| \notag \\
\leq& \frac{C\lambda}{N} \bigg[ \| \eta \|^2 \| ( \mathcal{N} + 1)^{(n+4)/2} e^{\lambda \mathcal{N}} \xi \| + \vert \eta (x;y) \vert \; \| \eta \| \; \| ( \mathcal{N} + 1)^{(n+4)/2} e^{\lambda \mathcal{N}} \xi \| \notag \\
& \quad \quad \quad + \| \eta \| \; \|a_x ( \mathcal{N} + 1)^{(n+3)/2} e^{\lambda \mathcal{N}} \xi \| + \| \eta \|^2 \| a_xa_y ( \mathcal{N} + 1)^{(n+2)/2} e^{\lambda \mathcal{N}} \xi \| \bigg]  \; , 
\end{align}
and  
\begin{align}
\| ( \mathcal{N} &  + 1)^{n/2}e^{-\lambda \mathcal{N}} \big[ e^{\lambda \mathcal{N}} \big[ e^{\lambda \mathcal{N}}, b_x d_y \big] \xi \| \notag \\
\leq& \frac{C\lambda^2}{N} \bigg[ \| \eta \|^2 \| ( \mathcal{N} + 1)^{(n+4)/2} e^{\lambda \mathcal{N}} \xi \| + \vert \eta (x;y) \vert \; \| \eta \| \; \| ( \mathcal{N} + 1)^{(n+4)/2} e^{\lambda \mathcal{N}} \xi \| \notag \\
& \quad \quad \quad + \| \eta \| \; \|a_x ( \mathcal{N} + 1)^{(n+3)/2} e^{\lambda \mathcal{N}} \xi \| + \| \eta \|^2 \| a_xa_y ( \mathcal{N} + 1)^{(n+2)/2} e^{\lambda \mathcal{N}} \xi \| \bigg]  \; \; , 
\end{align}
as well as 
\begin{align}
\|   ( \mathcal{N}  & + 1)^{n/2} \big[ e^{\lambda \mathcal{N}}, d_xd_y \big] \xi \| \notag \\
\leq& \frac{C \lambda}{N^2} \bigg[ \| \eta \|^2 \| ( \mathcal{N} + 1)^{n+6)/2} e^{\lambda \mathcal{N}} \xi \| +  \vert  \eta(x;y) \vert \; \| \eta \| \| ( \mathcal{N} + 1)^{(n+4)/2} e^{\lambda \mathcal{N}} \xi \| \notag \\
& \quad \quad \quad  + \| \eta \|^2 \| a_x ( \mathcal{N} +1)^{(n+5)/2} e^{\lambda \mathcal{N}} \xi \| + \| \eta \|^2 \| a_y \| ( \mathcal{N} + 1)^{(n+4)/2} e^{\lambda \mathcal{N}} \xi \| \notag \\
& \quad \quad \quad + \| \eta \|^2 \| a_xa_y( \mathcal{N} + 1)^{(n+4)/2} e^{\lambda \mathcal{N}} \xi \| \bigg] \; . 
\end{align}
and 
\begin{align}
\|   ( \mathcal{N}  & + 1)^{n/2} e^{-\lambda \mathcal{N}} \big[ e^{\lambda \mathcal{N}}, \; \big[ e^{\lambda \mathcal{N}}, d_xd_y \big] \xi \| \notag \\
\leq& \frac{C \lambda^2}{N^2} \bigg[ \| \eta \|^2 \| ( \mathcal{N} + 1)^{n+6)/2} e^{\lambda \mathcal{N}} \xi \| +  \vert  \eta(x;y) \vert \; \| \eta \| \| ( \mathcal{N} + 1)^{(n+4)/2} e^{\lambda \mathcal{N}} \xi \| \notag \\
& \quad \quad \quad  + \| \eta \|^2 \| a_x ( \mathcal{N} +1)^{(n+5)/2} e^{\lambda \mathcal{N}} \xi \| + \| \eta \|^2 \| a_y \| ( \mathcal{N} + 1)^{(n+4)/2} e^{\lambda \mathcal{N}} \xi \| \notag \\
& \quad \quad \quad + \| \eta \|^2 \| a_xa_y( \mathcal{N} + 1)^{(n+4)/2} e^{\lambda \mathcal{N}} \xi \| \bigg] \; . 
\end{align}
\end{lemma}

The proof of Lemma \ref{lemma:error-bounds2} follows from the arguments presented in \cite[Section 2, Lemma 2.4]{NR}. We remark that the only difference of \cite[Lemma 2.4]{NR} compared to Lemma \ref{lemma:error-bounds2} above is that \cite{NR} considers translation invariant systems on the unit torus and considers the error estimates on the corresponding discrete momentum space. Here, however, we work in the full space $\mathbb{R}^3$ and the error bounds above are formulated in position space. The arguments used in \cite[Lemma 2.4]{NR} in momentum space, however, apply in position space as well, up to minor technical modifications, and thus yield in Lemma \ref{lemma:error-bounds2}; we skip the details. 

We remark that later in the proof of Proposition \ref{prop:GN} in Section \ref{sec:GN}, we apply Lemma \ref{lemma:error-bounds2} for the choice of $\eta$ given by \eqref{eq:defeta}. In fact, by Lemma \ref{lm:bndseta}, we can then pick $\ell >0$ sufficiently small such that $\| \eta \|$ meets the assumptions of Lemma \ref{lemma:error-bounds2}.

Now we are ready to analyze all the single contributions of the renormalized excitation Hamiltonian $\mathcal{G}_N$ defined in \eqref{def:GN}. As mentioned earlier, the analysis of the different contributions $\cG_N^{(k)} $ is quite similar to that in \cite{BSS1}, up to minor technical modifications. Since the key arguments are standard by now, we focus here only on a detailed analysis of $\cG_N^{(0)}$. Afterwards, we summarize the key properties of $ \cG_N^{(1)}$ to $\cG_N^{(4)}$ without providing detailed proofs. On the one hand, we hope that this illustrates the involved arguments in sufficient detail and, on the other hand, we hope that this avoids the repetition of very similar arguments many times before, e.g.\ in \cite{BS, BBCS1, BBCS2, BBCS3, BBCS4, ABS, BSS1, BSS2}.

\subsection{Analysis of $\mathcal{G}_N^{(0)}$}\label{GN0-analysation}
\label{sec:GN0}

In this section we study $\mathcal{G}_N^{(0)} = e^{-B(\eta)} \mathcal{L}_N^{(0)} e^{B( \eta)}$ and recall for this the definition of $\mathcal{L}_N^{(0)}$ from \eqref{eq:cLNj}. We have 
\begin{align}
\mathcal{L}_N^{(0)}&= C_N+(C_{N,1}+C_{N,2})\mathcal{N}+C_{N,3}\frac{\mathcal{N}^2}{N} \label{decompositionLn0} 
\end{align}
with 
\begin{align}
C_N &:=N\left<\varphi_{\text{GP}}, [-\Delta+V_{\text{ext}}+\frac{1}{2}(N^3V(N\cdot)*\varphi_{\text{GP}}^2)]\varphi_{\text{GP}} \right> \nonumber\\
&\quad -\frac{1}{2}\left<\varphi_{\text{GP}},(N^3V(N\cdot)*\varphi_{\text{GP}}^2)\varphi_{\text{GP}}\right>, \label{CN-GN0-def}\\
C_{N,2} &:= - \left<\varphi_{\text{GP}}, [-\Delta+V_{\text{ext}}+\frac{1}{2}(N^3V(N\cdot)*\varphi_{\text{GP}}^2)]\varphi_{\text{GP}} \right>, \nonumber\\
C_{N,3} &:= -\left(1-\frac{1}{N}\right)\frac{1}{2}\left<\varphi_{\text{GP}},(N^3V(N\cdot)*\varphi_{\text{GP}}^2)\varphi_{\text{GP}}\right>, \nonumber\\
C_{N,4} &:= \frac{1}{2}\left<\varphi_{\text{GP}},(N^3V(N\cdot)*\varphi_{\text{GP}}^2)\varphi_{\text{GP}}\right>.\nonumber
\end{align}

The goal of this Section is to prove the following Proposition that shows that $\mathcal{G}_{N}^{(0)}$ satisfy all necessary properties for Proposition \ref{prop:GN}. 

\begin{proposition}
\label{prop:GN0-prop}
Under the same assumptions as in Proposition \ref{prop:GN}, we have 
\begin{align}
e^{-B(\eta)}\mathcal{L}_N^{(0)}e^{B(\eta)} = C_N+\mathcal{E}_N^{(0)}
\end{align} 
with $C_N$ defined in \eqref{CN-GN0-def} and where $\mathcal{E}_N^{(0)}$ satisfies for all ${\xi\in\mathcal{F}^{\leq N}}$ and $\kappa >0$ small enough 
\begin{align}
\big\vert \big\langle \xi, \mathcal{E}_N^{(0)} \xi\big\rangle\big\vert  & \leq  C  \;  \langle \xi, \; (\mathcal{N} + 1) \xi \rangle  \\
\big\vert \big\langle  \xi, e^{\kappa \mathcal{N}/2} \big[ e^{\kappa \mathcal{N}/2}, \mathcal{E}_N^{(0)} \big]  \xi\big\rangle \big\vert  & \leq C \kappa \;  \langle \xi, \; (\mathcal{N} + 1) e^{\kappa \mathcal{N}} \xi \rangle  
\end{align}
for some constant ${C>0}$.
\end{proposition}

\begin{proof}
The term $\mathcal{G}_N^{(0)}$ is the sum of multiples of powers of the number of particles operator $\mathcal{N} = \int dx \; a_x^*a_x  $ and thus the first bound of Proposition \ref{prop:GN0-prop} is an immediate consequence of the inequality proven for 
\begin{align} 
\label{eq:bogo-N}
e^{-B(\eta)}(\mathcal{N}+1)^ke^{B(\eta)}&\leq C^k (\mathcal{N}+1)^k
\end{align}
that holds on $\mathcal{F}^{\leq N}$  for any $k \in \mathbb{N}$,proven for example in \cite[Lemma 2.1]{BBCS3}. 

In order to prove the second bound, we start with considering the contribution of $\mathcal{G}_N^{(0)}$ that is a multiple of $\mathcal{N}$. Thus, for this term we need to analyze 
\begin{align}
\label{eq:comm1}
 e^{\kappa \mathcal{N}/2}\big[ e^{\kappa \mathcal{N}/2}, e^{-B( \eta)} \mathcal{N} e^{B( \eta)}  \big]  
\end{align}
and we start with computing the action of the generalized tranform on the number of particles operator. To this end, we write  
\begin{align}
 e^{-B( \eta)} \mathcal{N} e^{B( \eta)}  = \mathcal{N} + \int_0^1 ds \; e^{-sB( \eta)} \big[ B( \eta),  \mathcal{N} \big] e^{sB( \eta)}  
\end{align}
and use the modified CCR in \eqref{eq:CCR-mod} to compute the r.h.s.\ 
\begin{align}
  & e^{-B( \eta)} \mathcal{N} e^{B( \eta)}  \notag \\
 &= \mathcal{N} + \int_0^1 ds \; e^{-sB( \eta)}  \int dxdy \;  \eta(x;y)  \big( b_x^*b_y^* + b_xb_y \big) e^{sB( \eta)}     \; . 
\end{align}
With the Bogoliubov transform's action \eqref{eq:action-bogo-mod} we can further write the integrand as 
\begin{align}
  & e^{-B( \eta)} \mathcal{N} e^{B( \eta)} - \mathcal{N} \notag \\
 &= \int_0^1 ds \; \int dxdy \;  \eta(x;y)  \big( b^* ( \sinh_{s\eta,x} )  b^* ( \sinh_{s\eta,y} )   +b^*( \cosh_{s\eta,x} )  b^* ( \cosh_{s\eta,y} )  +  {\rm h.c.}  \big) \notag \\
 &+ 2 \int_0^1 ds  \;  \int dxdy \;  \eta(x;y)  \big( b^* ( \sinh_{s\eta,x} )  b ( \cosh_{s\eta,y} )   + b^*( \cosh_{s\eta,x} )  b ( \sinh_{s\eta,y} )  \big) \notag \\ 
 &+2 \int_0^1 ds \;  \int dxdy \;  \eta(x;y) \; \big(  \big[ b (   \cosh_{s\eta,y} ),  b^* (   \sinh_{s\eta,x} ) \big] + {\rm h.c.} \big) \notag \\
 &+  \int_0^1 ds \;  \int dxdy \;  \eta(x;y)  \big( b^* ( \sinh_{s\eta,x} ) + b ( \cosh_{s\eta,x} ) \big) d^*_{s\eta, y} + {\rm h.c.}   \notag \\
 &+  \int_0^1 ds \;  \int dxdy \;  \eta(x;y)  \;  d_{s\eta, y}^*  \big( b^* ( \sinh_{s\eta,x} ) + b ( \cosh_{s\eta,x} ) \big) + {\rm h.c.}  \notag \\
  &+  \int_0^1 ds \;  \int dxdy \;  \eta(x;y)  \;  d_{s\eta, y}^*  d_{s\eta,x}^* + {\rm h.c.}  
\end{align}
We recall that we need to control the commutator of $e^{-B( \eta)} \mathcal{N} e^{B( \eta)} $ with $e^{\kappa \mathcal{N}/2}$ that by the previous calculation reduces to computing the commutators of $b_y, b_x^*$ and $d_{s\eta,x}^*, d_{s\eta,x}$ with $e^{\kappa \mathcal{N}/2}$. While for the first, we use (similar as for the standard creation and annihilation operators in Section \ref{sec:outline}) that 
\begin{align}
\mathcal{N} b^*_x = b_x^* ( \mathcal{N} + 1), \quad \text{and} \quad \mathcal{N} b_x = b_x ( \mathcal{N} - 1) \label{eq:pull-through}, 
\end{align}
we use Lemma \ref{lemma:error-bounds2} to control the commutators with the errors $d_{s\eta,x}^*, d_{s, \eta,x}$. To be more precise, we write 
\begin{align}
& e^{\kappa \mathcal{N}/2}  \big[  e^{\kappa \mathcal{N}/2}, e^{-B( \eta)} \mathcal{N} e^{B( \eta)} \big]  \notag \\
=& \big( e^{\kappa} - 1 \big)e^{\kappa \mathcal{N}/2} \int_0^1 ds \notag \\
& \hspace{1cm} \times \; \int dxdy \;  \eta(x;y)  \big( b^* ( \sinh_{s\eta,x} )  b^* ( \sinh_{s\eta,y} )   +b^*( \cosh_{s\eta,x} )  b^* ( \cosh_{s\eta,y} )  \big) e^{\kappa \mathcal{N}/2} \notag \\
&+ \big( e^{-\kappa} - 1 \big) e^{\kappa \mathcal{N}/2} \int_0^1 ds  \notag \\
& \hspace{1cm} \times \; \int dxdy \;  \eta(x;y)  \big( b ( \sinh_{s\eta,x} )  b ( \sinh_{s\eta,y} )   +b( \cosh_{s\eta,x} )  b ( \cosh_{s\eta,y} )  \big) e^{\kappa \mathcal{N}/2} \notag \\
&+ e^{\kappa \mathcal{N}/2} \int_0^1 ds   \int dxdy \;  \eta(x;y)  \big( e^{ \kappa/2} b^* ( \sinh_{s\eta,x} ) +e^{ -\kappa/2}  b ( \cosh_{s\eta,x} ) \big) \big[ e^{\kappa \mathcal{N}/2}, d^*_{s\eta, y} \big] \notag \\
& \hspace{12cm}  - {\rm h.c.}   \notag \\
&+ e^{\kappa \mathcal{N}/2} \int_0^1 ds    \int dxdy \;  \eta(x;y)  \big( (e^{ \kappa/2}-1)  b^* ( \sinh_{s\eta,x} ) \notag \\
& \hspace{6cm}  +( e^{ -\kappa/2} -1)  b ( \cosh_{s\eta,x} ) \big)  d^*_{s\eta, y} e^{\kappa \mathcal{N}/2}    - {\rm h.c.}   \notag \\
 &+ e^{\kappa \mathcal{N}/2}\int_0^1 ds \; \int dxdy \;  \eta(x;y) \big[ e^{\kappa \mathcal{N}/2}, d^*_{s\eta, y} \big] \big( e^{ \kappa/2} b^* ( \sinh_{s\eta,x} ) +e^{ \kappa/2}  b ( \cosh_{s\eta,x} ) \big) \notag \\
 & \hspace{12cm}   - {\rm h.c.}   \notag \\
&+ e^{\kappa \mathcal{N}/2} \int_0^1 ds \;   \int dxdy \;  \eta(x;y) d^*_{s\eta, y}  \big( (e^{ \kappa/2}-1)  b^* ( \sinh_{s\eta,x} ) +( e^{- \kappa/2} -1)  b ( \cosh_{s\eta,x} ) \big)    \notag \\
& \hspace{12cm}  - {\rm h.c.}   \notag \\
 &+e^{\kappa \mathcal{N}/2} \int_0^1 ds \;  \int dxdy \;  \eta(x;y)  \;\big[ e^{\kappa \mathcal{N}/2},  d_{s\eta, y}^*  d_{s\eta,x}^* \big] - {\rm h.c.}    \label{eq:comm2}
\end{align}
and estimate in the following all terms of the r.h.s.\ of \eqref{eq:comm2} separately. We start with the first one for which we introduce the notation 
\begin{align}
H_{s} (w;z) := \eta (x;y) \big( \sinh_{\eta} (z;y) \sinh_{\eta} (w;x) + \cosh_{\eta} (z;y) \cosh_{\eta} (w;x) \big)  \; . 
\end{align}
Then, using that $\vert e^{\kappa} - 1 \vert \leq C \kappa $ for all $0 \leq \kappa \leq 1$, Eq. \eqref{eq:pull-through} and  the bound \eqref{eq:bounds-b2}, we find 
\begin{align}
\vert \langle \xi,  & \big( e^{\kappa} - 1 \big)e^{\kappa \mathcal{N}/2} \int_0^1 ds \; \int dzdw \;  H_s (z;w)  b_z^* b_w^*  e^{\kappa \mathcal{N}/2} \xi \rangle \vert  \notag \\
\leq& C \kappa \int_0^1 ds  \| H_s \|_2 \| (\mathcal{N} + 1)^{1/2} e^{\kappa \mathcal{N}/2} \xi \| \; \| ( \mathcal{N} + 1)^{1/2} e^{\kappa \mathcal{N}/2} \xi \| \; . 
\end{align}
Since 
\begin{align}
\label{eq:decomp-cosh}
\cosh_{\eta} (x;z) = \delta (x-z) + {\rm r}_{s\eta} (x;z), \quad \text{and} \quad \| {\rm r}_{s\eta} \|_2 \leq C 
\end{align}
 by Lemma \ref{lm:bndseta} for all $0 \leq s \leq 1$, we get 
\begin{align}
\| H_s \|_2 \leq C 
\end{align}
for a constant $C>0$ that does not depends on $s$ and arrive for the first term of the r.h.s.\ of \eqref{eq:comm2} at 
\begin{align}
\vert \langle \xi,  & \big( e^{\kappa} - 1 \big)e^{\kappa \mathcal{N}/2} \int_0^1 ds \; \int dzdw \;  H_s (z;w)  b_z^* b_w^*  e^{\kappa \mathcal{N}/2} \xi \rangle \vert  \notag \\
\leq& C \kappa \| (\mathcal{N} + 1)^{1/2} e^{\kappa \mathcal{N}/2} \xi \| \; \| ( \mathcal{N} + 1)^{1/2} e^{\kappa \mathcal{N}/2} \xi \|  
\end{align}
that is of the desired form. The second term of the r.h.s.\ of \eqref{eq:comm2} can be bounded similarly. 

For the third term we use the notation 
\begin{align}
h_{s,y} (z) = \int dx \sinh_{s \eta} (z;x) \eta (x;y), \quad \text{and} \quad g_{s,y} (z) = \int dx \cosh_{s \eta} (z;x) \eta (x;y)
\end{align}
and estimate 
\begin{align}
\vert \langle  & \xi, e^{\kappa \mathcal{N}/2}  \int_0^1 ds \;  \int dy \;   \big( e^{ \kappa/2} b^* ( h_{s,y} ) +e^{ -\kappa/2}  b ( g_{s,y} ) \big)  \big[ e^{\kappa \mathcal{N}/2}, d^*_{s\eta, y} \big]  \xi \rangle \notag \\
\leq&  \int_0^1 ds \int dy \; \| ( \mathcal{N} + 1)^{-1/2} \big(  b ( h_{s,y} ) +  b^* ( g_{s,y} ) \big) \big[ e^{\kappa \mathcal{N}/2}, d_{s\eta, y} \big] \xi \| \; \| ( \mathcal{N} + 1)^{1/2} e^{ \kappa \mathcal{N}/2} \xi \| \notag \\ 
&+  \int_0^1 ds \int dy \; \|  e^{-\kappa \mathcal{N}/2} \big[ e^{\kappa \mathcal{N}/2}, \big[ e^{\kappa \mathcal{N}/2}, d_{s\eta, y} \big]\big]  \xi \| \; \| \big( e^{ \kappa/2} b^* ( h_{s,y} ) +e^{ -\kappa/2}  b ( g_{s,y} ) \big)  e^{ \kappa \mathcal{N}/2} \xi \| 
\end{align}
With Lemma \ref{lemma:error-bounds2} and the bounds \eqref{eq:bounds-b1} we find 
\begin{align}
\vert \langle  & \xi, e^{\kappa \mathcal{N}/2}  \int_0^1 ds \;  \int dy \;   \big( e^{ \kappa/2} b^* ( h_{s,y} ) +e^{ -\kappa/2}  b ( g_{s,y} ) \big)  \big[ e^{\kappa \mathcal{N}/2}, d^*_{s\eta, y} \big]  \xi \rangle \notag \\
&\leq \frac{C\kappa}{N}  \int_0^1 ds \;  \int dy \; \big( \| h_{s,y} \| + \| g_{s,y} \| \big)  \bigg(  \| b_y ( \mathcal{N} + 1) e^{\kappa \mathcal{N}/2} \xi \|  + \| \eta_y \| \; \| ( \mathcal{N} + 1)^{3/2} e^{\kappa \mathcal{N}} \xi \| \bigg) \; \notag \\
& \hspace{4cm} \times \| ( \mathcal{N} + 1)^{1/2} e^{\kappa \mathcal{N}/2} \xi \|  \; . 
\end{align}
Recalling the decomposition \eqref{eq:decomp-cosh}, we have from Lemma \ref{lm:bndseta} 
\begin{align}
\| g_{s,y} \| \leq C \| \eta_y \|  , \quad \text{and} \quad \| h_{s,y}
 \| \leq C \| \eta_y \|
\end{align}
and $\| \eta_y \| \in L^2( \mathbb{R}^3)$ by Theorem \ref{thm:gp} so that we conclude since $\mathcal{N} \leq N$ on $\mathcal{F}_{\perp \varphi_{\rm GP}}^{\leq N}$ by Lemma \ref{lemma:error-bounds2} with the desired bound, that is 
\begin{align}
\vert \langle  &   \xi, e^{\kappa \mathcal{N}/2}  \int_0^1 ds  \int dy   \big( e^{ \kappa/2} b^* ( h_{s,y} ) +e^{ -\kappa/2}  b ( g_{s,y} ) \big)  \big[ e^{\kappa \mathcal{N}/2}, d^*_{s\eta, y} \big]  \xi \rangle \notag \\
\leq& C \kappa \langle \xi, ( \mathcal{N} + 1) \xi \rangle \; . 
\end{align}
The hermitian conjugate can be estimated simiarly and we proceed with the forth term of the r.h.s.\ of \eqref{eq:comm2}. With the same notation and similar ideas as before, get with Lemmas \ref{lemma:error-bounds1} 
\begin{align}
 \vert \langle \xi,  &  e^{\kappa \mathcal{N}/2} \int_0^1 ds \int dy \big( (e^{\kappa/2} - 1) b^*(h_{s,y} ) +   (e^{-\kappa/2} - 1)b(g_{s,y} ) ) d^*_{s\eta,y} e^{\kappa \mathcal{N}/2} \xi \rangle \vert \notag \\
 \leq& C \kappa   \| ( \mathcal{N} + 1)^{1/2} e^{\kappa \mathcal{N}/2} \xi \| \; \| ( \mathcal{N} + 1)^{1/2} e^{\kappa \mathcal{N}/2} \xi \| \; . 
\end{align}
A similar bound for the hermitian conjugate, as well as the fifth and sixth term of the r.h.s.\ of \eqref{eq:comm2} follow in the same way. For the seventh term we bound similarly
\begin{align}
\vert \langle \xi, &  e^{\kappa \mathcal{N}/2}   \int_0^1 ds \;  \int dxdy \;  \eta(x;y)  \;\big[ e^{\kappa \mathcal{N}/2},  d_{s\eta, y}^*  d_{s\eta,x}^* \big]\xi \rangle\vert  \notag \\
\leq& \int_0^1 ds \int dxdy \vert \eta (x;y) \vert \, \|e^{-\kappa \mathcal{N}/2} \big[ e^{\kappa \mathcal{N}/2},  d_{s\eta, y} d_{s\eta,x} \big] \xi \| \; \| e^{\kappa \mathcal{N}/2}  \xi \| \notag \\
\leq & C \kappa \langle \xi, ( \mathcal{N} + 1) e^{\kappa \mathcal{N}} \xi \rangle
\end{align}
and again the hermitian conjugate can be estimated with the same arguments. 
 
This finishes the proof of the desired bound for the first term of the r.h.s.\ of \eqref{eq:comm1}, and the bound for the second term of the r.h.s.\ can be bounded can be proven in the same way. Since $\vert C_{N,i}\vert \leq C $ by Theorem \ref{thm:gp}, the first term of $\mathcal{G}_N^{(0)}$, originating from the first term of $\mathcal{L}_{N}^{(0)}$, that is a multiple of $\mathcal{N}$, follows. The remaining term of $\mathcal{G}_N^{(0)}$, originating from the contribution of $\mathcal{L}_N^{(0)}$, that is quadratic in $\mathcal{N}$, can be estimated in the same way and we omit the details here. 
\end{proof}

\subsection{Analysis of $\mathcal{G}_N^{(1)}$}
\label{sec:GN1}

In this section we show that $\mathcal{G}_N^{(1)} = e^{- B( \eta)} \mathcal{L}_N^{(1)} e^{B (\eta)}$ satisfies Proposition \ref{prop:GN} and recall for this from \eqref{eq:cLNj} 
\begin{align}
\mathcal{L}_N^{(1)} &= \sqrt{N}b((N^3V(N\cdot)*\varphi_{\text{GP}}^2 -8\pi \mathfrak{a} \varphi_{\text{GP}}^2)\varphi_{\text{GP}})-\frac{\mathcal{N}+1}{\sqrt{N}}b((N^3V(N\cdot)*\varphi_{\text{GP}}^2)\varphi_{\text{GP}})+\text{h.c.}  
\end{align}
We furthermore introduce the notation 
\begin{align}
\label{def:hN}
h_N := \big( N^3 V( N \cdot ) \omega ( N \cdot ) * \varphi_{\text{GP}}^2 \big) \varphi_{\text{GP}} \; . 
\end{align} 

\begin{proposition}\label{prop:GN1-prop}
Under the same assumptions as in Proposition \ref{prop:GN}, we have 
\begin{equation}\label{eq:EN1-def}
 e^{-B(\eta)}\mathcal{L}_N^{(1)}e^{B(\eta)} = \sqrt{N} \left[b(\cosh_\eta(h_N))+b^*(\sinh_\eta(h_N))+\text{h.c.} \right] + \mathcal{E}_N^{(1)} 
\end{equation}
with $h_N$ given by \eqref{def:hN} and where $\mathcal{E}_N^{(1)}$ satisfies, for all ${\xi\in\mathcal{F}^{\leq N}}$,
\begin{align}
\big\vert \big\langle \xi, \mathcal{E}_N^{(1)} \xi\big\rangle\big\vert  & \leq  C  \;  \langle \xi, \; (\mathcal{N} + 1) \xi \rangle  \\
\big\vert \big\langle  \xi, e^{\kappa \mathcal{N}/2}\big[ e^{\kappa \mathcal{N}/2}, \mathcal{E}_N^{(1)} \big]  \xi\big\rangle \big\vert  & \leq C \kappa \;  \langle \xi, \; (\mathcal{N} + 1) e^{\kappa \mathcal{N}} \xi \rangle  \end{align}
for some constant ${C>0}$.

\end{proposition}

We remark that the terms $\sqrt{N} \left[b(\cosh_\eta(h_N))+b^*(\sinh_\eta(h_N))+\text{h.c.} \right]$ which we isolated from $\mathcal{G}_N^{(1)}$ will lead to cancellations with some terms from $\mathcal{G}_N^{(3)}$ (see Proposition~\ref{prop:GN3-prop} below) that are crucial in the Gross-Pitaevskii regime. 

The proof of Proposition \ref{prop:GN1-prop} relies, as the proof of Proposition \ref{prop:GN0-prop} in the previous Section, on the approximate action of the generalized Bogoloiubov transform \eqref{eq:action-bogo-mod} and the error estimates of the errors $d_{\eta,x}^*, d_{\eta,x}$ in Lemma \ref{lemma:error-bounds1} resp. their commutators with $e^{\kappa \mathcal{N}}$ in Lemma \ref{lemma:error-bounds2}. Straight forward arguments, as presented in the last Section then yield the desired bounds of Proposition \ref{prop:GN1-prop}. We omit further details and refer to \cite[Section 4.2]{BSS1} for similar arguments to derive the first bound resp. \cite[Section 3.2]{NR} to derive the second bound of Proposition \ref{prop:GN2-prop}.

\subsection{Analysis of $\mathcal{G}_N^{(2)}$}

In this section we study $\mathcal{G}_N^{(2)} = e^{- B( \eta)} \mathcal{L}_N^{(2)} e^{B(\eta)}$ through analyzing the three contributions 
\begin{align}
\mathcal{L}_N^{(2)} - \mathcal{K}- \mathcal{V}_{\text{ext}}  &=  \int dxdy N^3V(N(x-y))|\varphi_{\text{GP}}(y)|^2\left(b_x^*b_x-\frac{1}{N}a_x^*a_x\right) \nonumber\\
&\quad +\int dxdy N^3V(N(x-y))\varphi_{\text{GP}}(x)\varphi_{\text{GP}}(y)\left(b_x^*b_y-\frac{1}{N}a_x^*a_y \right) \nonumber\\
&\quad + \frac{1}{2}\int dxdy N^3V(N(x-y))\varphi_{\text{GP}}(y)\varphi_{\text{GP}}(x)(b_x^*b_y^*+\text{h.c.})
\end{align}
and 
\begin{align}
\mathcal{K} = \int dx  \,\nabla_xa_x^*\nabla_x a_x , \quad \text{and} \quad \mathcal{V}_{\text{ext}} = \int dx\, V_{\text{ext}}(x)a_x^*a_x  \; .  
\end{align}
For this, we introduce the notation 
\begin{align}
e^{-B(\eta)} (\mathcal{L}_N^{(2)}-\mathcal{K}-\mathcal{V}_{\text{ext}})e^{B(\eta)}
=& \int dxdy\, N^3V(N(x-y))\varphi_{\text{GP}}(x)\varphi_{\text{GP}}(y)\eta(x,y)  \notag \\
&+\frac{1}{2}\int dxdy\, N^3V(N(x-y))[\varphi_{\text{GP}}(x)\varphi_{\text{GP}}(y)b_x^*b_y^*+\text{h.c.}] \notag \\
&+\mathcal{E}_N^{(2)} \; . 
\end{align}
and 
\begin{align}
e^{-B(\eta)}\mathcal{K}e^{B(\eta)}
&=\int dxdy\, N^3(\Delta w)(N(x-y))[\varphi_{\text{GP}}(x)\varphi_{\text{GP}}(y)b_x^*b_y^*+\text{h.c.}] \nonumber\\
&\quad +\mathcal{K} +||\nabla_x\eta||^2+\mathcal{E}_N^{(K)},  \label{EN(K)-def} \notag \\
e^{-B(\eta)}\mathcal{V}_{\text{ext}}e^{B(\eta)} &= \mathcal{V}_{\text{ext}}+\mathcal{E}_N^{(V_{\text{ext}})} \; . 
\end{align}

The terms isolated from all the three contributions will cancel with terms from the quartic term (see Proposition \ref{prop:GN3-prop}), while the errors $\mathcal{E}_N^{(2)},\mathcal{E}_N^{(K)} $ and $\mathcal{E}_N^{(V_{\text{ext}})}$ satisfy the desired bounds of Proposition \ref{prop:GN}, as the following Proposition shows.   


\begin{proposition}
\label{prop:GN2-prop}
Under the same assumptions as in Proposition \ref{prop:GN}, the errors $\mathcal{E}_N^{(V_{\text{ext}})},\mathcal{E}_N^{(2)}$ and $\mathcal{E}_N^{(K)}$ satisfy for all ${\xi\in\mathcal{F}^{\leq N}}_{\varphi_\text{GP}} $ 
\begin{align}
\big\vert \big\langle \xi, \mathcal{E}_N^{(V_{\rm ext})} \xi\big\rangle\big\vert, \; \;  \big\vert \big\langle \xi, \mathcal{E}_N^{(K)} \xi\big\rangle\big\vert, \; \; \big\vert \big\langle \xi, \mathcal{E}_N^{(2)} \xi\big\rangle\big\vert  & \leq  C  \;  \langle \xi, \; (\mathcal{H}_N + \mathcal{N} + 1) \xi \rangle 
\end{align}
and 
\begin{align}
\big\vert \big\langle  \xi,  & e^{\kappa \mathcal{N}/2}\big[ e^{\kappa \mathcal{N}/2}, \mathcal{E}_N^{(V_{\rm ext})} \big]  \xi\big\rangle \big\vert, \, \; \big\vert \big\langle  \xi, e^{\kappa \mathcal{N}/2}\big[ e^{\kappa \mathcal{N}/2}, \mathcal{E}_N^{(K)} \big]  \xi\big\rangle \big\vert, \, \; \big\vert \big\langle  \xi, e^{\kappa \mathcal{N}/2}\big[ e^{\kappa \mathcal{N}/2}, \mathcal{E}_N^{(2)} \big]  \xi\big\rangle \big\vert  \notag \\
& \leq C \kappa \;  \langle \xi, \; (\mathcal{H}_N+ \mathcal{N} + 1) e^{\kappa \mathcal{N}} \xi \rangle  
\end{align}
for some ${C>0}$.
\end{proposition}

The Proposition follows with similar arguments as outlined in the proof of Proposition \ref{prop:GN0-prop} in Section \ref{sec:GN0} and builds on the ideas from \cite{BSS1,BSS2} and their extension to the exponential arguments in \cite{NR}. As the analysis closely follows \cite{BSS1,BSS2,NR} we omit the details here.

\subsection{Analysis of $\mathcal{G}_N^{(3)}$ and $\mathcal{G}_N^{(4)}$} \label{sec:GN3}

In this section we the analyze $\mathcal{G}_N^{(3)} = e^{-B( \eta)} \mathcal{L}_N^{(3)} e^{B( \eta)}$ and $\mathcal{G}_N^{(4)} = e^{-B( \eta)} \mathcal{L}_N^{(4)} e^{B( \eta)}$.  We write recalling from \eqref{def:hN} the definition $h_N = (N^3V(N\cdot)w(N\cdot)*\varphi_{\text{GP}}^2)\varphi_{\text{GP}}$ the cubic term as 
\begin{equation}
\label{EN3-def}
 e^{-B(\eta)}\mathcal{L}_N^{(3)} e^{B(\eta)} = -\sqrt{N}\left[ b(\cosh_\eta(h_N)+b^*(\sinh_\eta(h_N)+\text{h.c.} \right] + \mathcal{E}_N^{(3)}, 
\end{equation}
and the quartic term as 
\begin{align}
e^{-B(\eta)}\mathcal{L}_N^{(4)} e^{B(\eta)} &= \mathcal{V}_N +\frac{1}{2}\int dxdy\, N^2V(N(x-y))|\eta(x,y)|^2 \nonumber\\
&\quad + \frac{1}{2}\int dxdy\, N^2V(N(x-y))[k(x,y)b_x^*b_y^*+\text{h.c.}] \nonumber\\
&\quad +\mathcal{E}_N^{(4)}, \label{EN4-def}
\end{align} 

The isolated terms from \eqref{EN3-def} and \eqref{EN4-def} will cancel together with the terms isolated from $\mathcal{G}_N^{(0)}, \mathcal{G}_N^{(1)}$ resp. from $\mathcal{G}_N^{(2)}$ in \eqref{prop:GN2-prop} while the errors $\mathcal{E}_N^{(3)}, \mathcal{E}_N^{(4)} $ satisfy the desired bounds as summarized in the following Proposition.

\begin{proposition}\label{prop:GN3-prop}
Under the same assumptions as in Proposition \ref{prop:GN}, the errors $\mathcal{E}_N^{(3)}$ and $\mathcal{E}_N^{(4)}$ satisfy for all ${\xi\in\mathcal{F}^{\leq N}} $ 
\begin{align}
\big\vert \big\langle \xi, \mathcal{E}_N^{(3)} \xi\big\rangle\big\vert, \; \;  \big\vert \big\langle \xi, \mathcal{E}_N^{(4)} \xi\big\rangle\big\vert   & \leq  C  \;  \langle \xi, \; (\mathcal{H}_N+ \mathcal{N} + 1) \xi \rangle  \\
\big\vert \big\langle  \xi, e^{\kappa \mathcal{N}/2} \big[ e^{\kappa \mathcal{N}/2}, \mathcal{E}_N^{(3)} \big]  \xi\big\rangle \big\vert, \, \; \big\vert \big\langle  \xi, e^{\kappa \mathcal{N}/2} \big[ e^{\kappa \mathcal{N}/2}, \mathcal{E}_N^{(4)} \big]  \xi\big\rangle \big\vert  & \leq C \kappa \;  \langle \xi, \; (\mathcal{H}_N + \mathcal{N} + 1) e^{\kappa \mathcal{N}} \xi \rangle  
\end{align}
for some ${C>0}$. 
\end{proposition}

We again skip the details of the proof of Proposition \ref{prop:GN3-prop} that follows with similar ideas as in \cite{BSS1,BSS2} and \cite{NR}and outlined in Section \ref{sec:GN0}.

\subsection{Proof of Proposition \ref{prop:GN}}
\label{ProofOfGNGoal}

In this section we prove Proposition \ref{prop:GN} from Proposition \ref{prop:GN0-prop} - \ref{prop:GN3-prop}. 

\begin{proof}[Proof of Proposition~\ref{prop:GN}]

Summarizing, we find from Proposition~\ref{prop:GN0-prop}  - \ref{prop:GN3-prop} that 
\begin{align}
\mathcal{G}_N &= C_N+\mathcal{H}_N^{\text{trap}}+\widetilde{\mathcal{E}}_N \notag \\
&\quad+ \int dxdy\,N^3(\Delta w_\ell)(N(x-y)) \varphi_{\text{GP}}(x)\varphi_{\text{GP}}(y)[b_x^*b_y^*+\text{h.c.}] \notag \\
&\quad + \frac{1}{2}\int dxdy\, N^3V(N(x-y))[\varphi_{\text{GP}}(x)\varphi_{\text{GP}}(y)b_x^*b_y^*+\text{h.c.}]\notag \\
&\quad + \frac{1}{2} \int dxdy \, N^2V(N(x-y))[k(x,y)b_x^*b_y^*+\text{h.c.}]   \label{eq:alles-eins}
\end{align}
where the constant $C_N$ is given by 
\begin{align}
C_N &=||\nabla_x\eta||^2-\frac{1}{2}\left<\varphi_{\text{GP}},(N^3V(N\cdot)*|\varphi_{\text{GP}}|^2)\varphi_{\text{GP}}\right> \notag \\ &\quad +N\left<\varphi_{\text{GP}}, \left(-\Delta+V_{\text{ext}}+\frac{1}{2}(N^3V(N\cdot)*|\varphi_{\text{GP}}|^2)\right)\varphi_{\text{GP}} \right>  \notag \\
&\quad +\int dxdy\, N^3V(N(x-y))\varphi_{\text{GP}}(x)\varphi_{\text{GP}}(y)\eta(x,y) \notag \\
&\quad + \frac{1}{2}\int dxdy\, N^2V(N(x-y))|\eta(x,y)|^2   \\
&= N\cE_\text{GP}(\ph_\text{GP}) + O(1)
\end{align}
and the error $\widetilde{\mathcal{E}_N}$ satisfies, for all ${\xi\in\mathcal{F}^{\leq N}}$,
\begin{equation}\label{EN-bound1}
\big|\big<\xi,\widetilde{\mathcal{E}}_N\xi\big>\big| \leq C \langle \xi, ( \mathcal{H}_N + \mathcal{N} + 1) \xi \rangle   
\end{equation}
and 
\begin{equation}\label{EN-bound2}
\big|\big<\xi, e^{\frac\kappa2 \mathcal{N}}  \big[ \widetilde{\mathcal{E}}_N, e^{\frac\kappa2 \mathcal{N}} \big] \xi\big>\big| \leq C \kappa \langle \xi, ( \mathcal{N} + \mathcal{H}_N + 1) e^{\kappa \mathcal{N}} \xi \rangle   \; . 
\end{equation}
The second bound is exactly the second bound \eqref{eq:GNbnd2} of Proposition \ref{prop:GN}. 

To prove the first bound \eqref{eq:GNbnd1} of Proposition \ref{eq:GNbnd1},  we note that the last three terms of the r.h.s.\ of \eqref{eq:alles-eins} sum up to 
\begin{align}
 N &\int dxdy\, \left( -\Delta +\frac{1}{2}N^2V(N(x-y)) \right)f_\ell(N(x-y)) \varphi_{\text{GP}}(x)\varphi_{\text{GP}}(y)b_x^*b_y^*  +\text{h.c.}  \label{eq:A}
\end{align}
and since $f_\ell$ solves the scattering equation, this term is bounded by $C( \cN_{\bot \ph_\text{GP}}+1)$. From \eqref{EN-bound1} and Cauchy Schwarz, we get that there exists $C>0$ such that 
\begin{align}
 \widetilde{\mathcal{E}}_N \geq -\frac{1}{2}\mathcal{H}_N-C(\mathcal{N}+1)
\end{align}
and therefore 
\begin{equation}\label{GN-CN-inequality}
 \mathcal{G}_N - N\cE_\text{GP}(\ph_\text{GP})   \geq \frac{1}{2}\mathcal{H}_N-C(\mathcal{N}+1).
\end{equation}
By \eqref{eq:coerc}, we furthermore know that there exists $C >0$ such that 
\begin{align}
\label{eq:GN1}
\mathcal{G}_N-N\cE_\text{GP}(\ph_\text{GP})  \geq  C^{-1} \,\mathcal{N}-C
\end{align} 
so that a basic interpolation between the last two bounds shows that 
\begin{align}
\label{eq:GN2}
\mathcal{G}_N-N\cE_\text{GP}(\ph_\text{GP}) 
\geq  \varepsilon ( \mathcal{H}_N+\cN) - C 
\end{align}
for some  ${\varepsilon >0}$. This proves the first bound \eqref{eq:GNbnd1} of Proposition \ref{prop:GN}. 
\end{proof}

\appendix 

\section{Properties of the Gross-Pitaevskii Functional} \label{app:GP}
The following result collects useful properties of the GP functional $\cE_\text{GP}$ and its minimizer $\pn$. For its proof, see \cite[Theorems 2.1, 2.5 \& Lemma A.6]{LSY} and \cite[Appendix A]{BSS1}. In view of \eqref{eq:expdecaypn} below, we remind the reader that we assume throughout this paper the assumptions \eqref{ass:V} $(ii)$ on $V_\text{ext}$.  

\begin{theorem} \label{thm:gp}
	Let $\cD_\emph{GP} = \{\psi \in L^2(\bR^3): \psi \in H^1(\bR^3), \psi \in L^2(\bR^3, V_\emph{ext}(x) dx ) \}$. Then, there exists a minimizer $\pn\in \cD_\emph{GP}$ with $ \|\pn\|_2=1$ such that
			\[ \inf_{\psi\in \mathcal{D}_\emph{GP} \ : \|\psi\|_2 =1} \mathcal{E}_\text{GP}(\psi) = \mathcal{E}_\text{GP}(\pn).  \]
	The minimizer $\ph_\emph{GP}$ is unique up to a complex phase, which can be chosen so that $\ph_\emph{GP}$ is strictly positive. Furthermore, the minimizer $\ph_\emph{GP}$ solves the Gross-Pitaevskii equation	
	\begin{equation} \label{eq:gpeq}
	-\Delta\ph_\emph{GP} + V_\emph{ext} \ph_\emph{GP} + 8\pi \mathfrak{a} \vert \ph_\emph{GP} \vert^2 \ph_\emph{GP} = \eps_\emph{GP}\ph_\emph{GP},
	\end{equation}
where $ \eps_\emph{GP}= \mathcal{E}_\text{GP}(\ph_\emph{GP}) + 4\pi \mathfrak{a} \Vert \ph_\emph{GP} \Vert_4^4.$
	
Finally, $\pn\in H^2(\mathbb{R}^3)\cap C^2(\mathbb{R}^3)$ and for every $\nu>0$, there exists a constant $C_\nu$ such that for all $x\in \bR^3$ it holds true that
	\begin{equation} \begin{split} \label{eq:expdecaypn}
	\vert \ph_\emph{GP}(x)\vert , \vert \nabla \ph_\emph{GP} (x) \vert,  \vert \Delta \ph_\emph{GP} (x) \vert  & \leq C_\nu  e^{-\nu \vert x \vert}.
	\end{split}
	\end{equation} 
	
\end{theorem}


\medskip
\noindent\textbf{Acknowledgements.} C. B. acknowledges support by the Deutsche Forschungsgemeinschaft (DFG, German Research Foundation) under Germany’s Excellence Strategy – GZ 2047/1, Projekt-ID 390685813. SR is supported by the European Research Council via the ERC CoG RAMBAS - Project - Nr. 10104424.






\end{document}